\documentclass[12pt]{article}
\usepackage[utf8]{inputenc}
\usepackage{geometry}
\usepackage{amsmath}
\usepackage{amsthm}
\usepackage{thmtools} %

\newtheorem{theorem}{Theorem}[section]
\newtheorem{lemma}[theorem]{Lemma}
\newtheorem{proposition}[theorem]{Proposition}
\newtheorem{definition}[theorem]{Definition}
\newtheorem{example}[theorem]{Example}

\newtheorem*{excont}{Example \protect\continuation}
\newcommand{\continuation}{??}
\newenvironment{examplecont}[1]{
  \renewcommand{\continuation}{\ref{#1}}
  \excont[Continued]
}{\endexcont}

\usepackage[ttscale=.875]{libertine}
\usepackage[scaled=0.785]{beramono} %
\usepackage[libertine]{newtxmath}
\usepackage[T1]{fontenc}

\allowdisplaybreaks

\usepackage{xifthen}
\usepackage{xcolor}
\usepackage{graphicx}
\usepackage{nicefrac}
\usepackage{subcaption}
\usepackage{accents}
\usepackage[normalem]{ulem}
\usepackage{cancel}
\usepackage{enumitem}
\usepackage{siunitx}
\usepackage{physics}
\usepackage{dirtytalk}
\usepackage{booktabs}
\usepackage{longtable}

\usepackage{math_macros}

\newcommand{\ubar}[1]{\underaccent{\bar}{#1}}

\usepackage{stackengine}
\newcommand\ubaralt[1]{\stackunder[1pt]{$#1$}{\rule{.6ex}{.1ex}}}

\usepackage{hyperref}
\usepackage{cleveref}
\usepackage[authoryear]{natbib}

\hypersetup{
    colorlinks=true,
    citecolor=teal
}

\newcommand{\newterm}[1]{\textit{#1}}

\title{\bf Price Cycles in Ridesharing Platforms}
\author{Chenkai Yu\thanks{Columbia University, NY, USA, \texttt{cyu26@gsb.columbia.edu}} \qquad Hongyao Ma\thanks{Columbia University, NY, USA, \texttt{hongyao.ma@columbia.edu}} \qquad Adam Wierman\thanks{Caltech, CA, USA, \texttt{adamw@caltech.edu}}}
\date{}

\begin{document}

\maketitle

\begin{abstract}
  In ridesharing platforms such as Uber and Lyft, it is observed that drivers sometimes collaboratively go offline when the price is low, and then return after the price has risen due to the perceived lack of supply. This collective strategy leads to cyclic fluctuations in prices and available drivers, resulting in poor reliability and social welfare. We study a continuous time, non-atomic model and prove that such online/offline strategies may form a Nash equilibrium among drivers, but lead to a lower total driver payoff if the market is sufficiently dense. Further, we show how to set price floors that effectively mitigate the emergence and impact of price cycles.
\end{abstract}

\section{Introduction} \label{sec:intro}

Ridesharing platforms have seen great success in recent years. Compared with traditional taxi systems, these platforms significantly increase the fraction of time when drivers have a rider in the back seat~\citep{cramer2016disruptive}.
In addition to mobile apps, which enable more efficient matching between riders and drivers, platforms such as Uber and Lyft also employ dynamic ``surge'' pricing to balance supply and demand in real-time.
This guarantees that rider wait times do not exceed a few minutes~\citep{rayle2014app}, achieving a desirable trade-off between (i) the cost of maintaining open driver supply, and (ii) the cost of long waiting time for riders and pick-up time for drivers~\citep{castillo2017surge,yan2020dynamic}.

In the presence of surge pricing, and with the ``real-time flexibility'' to decide when and where to drive~\citep{hall2016analysis,chen2019value}, drivers often strategically optimize to increase their earnings instead of accepting all trip dispatches from the platforms.
For example, many drivers cherry-pick based on the lengths or the destinations of trips~\citep{garg2020driver,castro2021randomized,changeRuleAdjustStrategy}. When prices fail to be sufficiently \emph{smooth} in space and time, drivers also decline trips, chasing surge prices in neighboring areas, or going off-line before large events will end in anticipation of a price hike~\citep{ma2019spatio}.

Increasingly, the strategic behavior of drivers is not limited to individuals. Collective strategic behavior of groups of drivers supplying the same region has been reported in recent years. A prominent collective strategy for drivers is to all go offline at the same time and return only after the price surges. For example, as reported by ABC News~\citep{WJLA2019ReganNationalCollusion,dustinTalk}:
\begin{quote}
  Every night, several times a night, Uber and Lyft drivers at Reagan National Airport simultaneously turn off their ride share apps for a minute or two to trick the app into thinking there are no drivers available---creating a price surge. When the fare goes high enough, the drivers turn their apps back on and lock into the higher fare.
\end{quote}

Such online/offline strategies have also been discussed in online forums by drivers from many cities. For example, drivers discussed in the UberPeople London forum~\citep{uberPeopleLondon}
\begin{quote}
  \say{Guys stay logged off until surge} because \say{Less supply high demand = surge.}
\end{quote}
As another example, when discussing the strategy of leaving the app off until it suges 2.5x at the Los Angeles International Airport, a driver posted~\citep{uberPeopleLA}
\begin{quote}
  \say{Can someone create a huge sign with a posterboard that says \say{stay logged off until 2.5x}. Like people taking turns holding the sign right at the entrance...}
\end{quote}

When adopted by drivers supplying the same region, these strategies induce cyclic fluctuations in trip prices and the number of available drivers, i.e., \emph{price cycles}. 
Specifically, after all drivers have gone offline, the platform will gradually increase trip prices due to the perceived lack of supply.  Once the price has risen high enough, the drivers (including those who arrived in the region during the price hike) sign back on at the same time, which leads the trip price to fall.  Eventually, the price falls enough to prompt the drivers in the region to sign offline again.%
\footnote{In practice, instead of letting prices jump from base fare to maximum surge, platforms usually limit the rate of price changes for various reasons. Erratic prices lead to poor experiences for riders~\citep{dholakia2015everyone}. Moreover, in addition to balancing supply and demand, surge pricing also provides incentives for drivers to reposition themselves efficiently. Driving from one region to another takes time; thus prices that change too fast are less effective for inducing driver movement~\citep{uberDriverSurge,garg2020driver}.
}

Such price cycles undercut the platforms' mission of providing reliable transportation to riders~\citep{ubermissionChicago,lyftmission} --- when prices are increasing and drivers are not accepting dispatches, 
even riders with a high willingness-to-pay may be unable get access to reliable service. Further, from the drivers' perspective, while this behavior is sometimes described as drivers ``collectively gaming the system'' in order to ``regain their autonomy'' (see e.g.,  \citet{mohlmann2017hands}), it is not clear whether it leads to increased %
total utility for %
drivers in practice.

Despite attention from the press, the emergence of price cycles has not been studied from the perspective of the design and operation of ridesharing platforms. 
In this paper, we initiate such a study and seek insight into questions such as:
\begin{enumerate}[label=(\roman*)] \it
    \item Is it an equilibrium for all drivers to adopt these online/offline strategies? %
    \item Do drivers collectively benefit from such strategies?
    \item What are the impacts of price cycles on riders and the efficiency of the platform?
    \item How can platforms avoid the emergence and reduce the impact of price cycles?
\end{enumerate}

Addressing these questions is challenging due to the fact that the emergent price cycles are inherently dynamic, even when the market condition is stationary over time. 
Understanding price cycles and the corresponding market equilibria thereby requires characterizing drivers' best-response in non-stationary, dynamic settings. %
More broadly, robust control in the presence of strategic participants is relevant and important for settings beyond ridesharing. 
As an example, it is reported that many drivers on the food delivery platform DoorDash have been collectively declining any delivery job, until the platform raises 
driver payment for the job to at least \$7~\citep{doordashGaming}.

\subsection{Contributions}

In this paper, we answer the four questions highlighted above using a continuous time
model of pricing and matching by a ridesharing platform in a specific region (e.g., a city center).\footnote{The focus on trips originating from a specific region, modeling scenarios where price cycles emerge in practice (e.g. when many riders leave restaurants and bars in the city center at the end of the evening).}
We show that it may be a Nash equilibrium among drivers to collectively withhold supply %
and accept trips only after prices have risen. 
For markets that are sufficiently dense, however, we prove that the resulting price cycles \emph{reduce}, instead of improve, the total payoff of drivers. %
To mitigate the adverse impact on driver payoffs, as well as on social-welfare, we show how a platform can introduce ``price floors'' (i.e. lower bounds on trip prices) and effectively prevent the emergence of stable price cycles.

\paragraph{The Model}

We work in a non-atomic model where the arrival of driver supply and rider demand are stationary.
After arriving in the region, drivers join the pool of available drivers, and may (at any point of time) decide whether to stay offline or to remain online and accept trip dispatches. 
Upon arrival, a rider who requests a trip is matched to the closest online driver in the region. A higher number of online drivers increases the density of drivers in space and reduces the average \emph{en route time}, i.e. the waiting time for the riders / pick-up time for the drivers.

Waiting is costly for both riders and drivers.
The platform aims to optimize \emph{social welfare} in steady state, i.e., total rider value minus the total waiting costs incurred by riders and drivers.%
\footnote{
We focus in this paper on the optimization of social welfare instead of revenue. With the substantial network effect as well as the existing fierce competition, major platforms such as Uber and Lyft have been prioritizing growth instead of profit, and welfare optimization is aligned with this goal. %
} 
The socially optimal steady state (SOSS) balances the cost of open driver supply and the costs of long en route times. 
After a shock on the supply and/or the demand side that pushes the system away from the SOSS, the platform sets trip prices in order to guide the system back to the SOSS quickly subject to the constraint that prices must adjust smoothly over time.
In this setting, staying online and accepting dispatches at all times does not form an equilibrium among the drivers: when the platform aggressively reduces prices in order to use up excess drivers, for example, a driver may benefit from going offline and waiting for a better price.

\paragraph{Main Results}

We analyze a family of \emph{online/offline threshold strategies} on which drivers coordinate in practice~\citep{uberPeopleLA,WJLA2019ReganNationalCollusion}, i.e., strategies where drivers go offline when prices drop below some threshold and return online after the price has risen above some target level. 
We illustrate the cyclic fluctuations of trip prices and the number of online drivers induced by such strategies. %
We also provide sufficient conditions under which a price cycle is \emph{stable}, meaning that the cycle repeats itself over time and that the online/offline strategy forms a Nash equilibrium among the drivers (Theorem~\ref{thm:suf_cond}).

For markets that are sufficiently dense, however, every stable price cycle leads to a lower total driver payoff in comparison to the intended SOSS outcome (Theorem~\ref{thm:lower_payoff}).
This is because despite the fact that some drivers are paid higher than the SOSS price, many other drivers are paid lower amounts over the duration of one price cycle. In aggregate, we prove that drivers get a lower payoff from the trips, and at the same time incur a higher total waiting cost. %
Via examples, we illustrate the substantially decreased welfare and driver surplus under price cycles. Riders, on the other hand, may collectively get a higher surplus than that under the SOSS, since the price paid by riders who are picked-up can be lower on average.

A natural approach for reducing the emergence and impact of price cycles is to introduce a \emph{price floor}. Intuitively, setting a lower bound on trip prices prevents drivers' payoffs from dropping too low, thereby reducing drivers' incentives to go offline and wait for a better price.
We characterizes the market conditions under which price floors are effective, and provide the set of price floors that ensure no stable price cycles exist (Theorem~\ref{thm:platform_policy}). A simple rule of thumb from our analysis is that a platform can eliminate stable cycles by imposing a lower bound on prices at the point where rider demand doubles in comparison to that under the SOSS.

\subsection{Related Work} \label{sec:related_work}

The literature on ridesharing platforms is rapidly growing, covering both empirical studies of current platforms and theoretical analysis of market designs. Empirical work has studied  a wide range of topics, from consumer surplus~\citep{cohen2016using,castillo2020benefits} and the labor market of %
drivers~\citep{hall2017labor,hall2016analysis} to %
the %
flexible work arrangements~\citep{chen2019value,chen2020reservation,xu2020empirical}, among other topics. %
A variety of topics have also received attention from the theoretical modeling community, including work analyzing the optimal growth of two-sided platforms~\citep{lian2021optimal, fang2019prices}, competition between platforms~\citep{lian2021larger,DBLP:conf/wine/AhmadinejadNSSS19,fang2020loyalty}, operations in the presence of %
autonomous vehicles~
\citep{ostrovsky2019carpooling,lian2020autonomous}, and utilization-based minimum wage regulations~\citep{asadpour2019minimum}.

This paper connects to the literature on pricing and matching in ridesharing platforms. \citet{bimpikis2019spatial} and \citet{besbes2020surge} study revenue-optimal pricing when supply and demand are imbalanced in space. \citet{ashlagi2018maximum}, \citet{dickerson2018allocation} and \citet{aouad2020dynamic} focus on matching between riders and drivers and the pooling of shared rides, taking into consideration the online arrival of supply and demand in space. Further, \citet{kanoria2020blind}, \citet{qin2020ride} and \citet{ozkan2020dynamic} design policies that dispatch drivers from areas with relatively abundant supply, while \citet{cai2019role} and \citet{pang2017efficiency} look at the role of information availability and transparency in platform design.
Dynamic pricing~\citep{BanJohRiq2015Pricing}, state-dependent dispatching~\citep{banerjee2018state,castro2020matching}, driver admission control~\citep{afeche2018ride} and capacity planning~\citep{besbes2018spatial} are also studied using queueing-theoretical models.

In regard to dynamic ``surge'' pricing, \citet{castillo2017surge} and \citet{yan2020dynamic} demonstrate the importance of dynamic pricing in maintaining the spatial density of open driver supply, which improves operational efficiency by reducing waiting times for riders / pick-up times for drivers. Further, empirical studies have demonstrated the effectiveness of dynamic surge pricing in improving reliability and efficiency~\citep{hall2015effects}, increasing driver supply during high-demand times~\citep{chen2015dynamic}, as well as creating incentives for drivers to relocate to higher surge areas~\citep{lu2018surge}.
In recent work, \citet{freund2021pricing} discuss the cycles of volatile driver supply, which lead to an increased average waiting times for riders and pick-up costs for drivers. %
These cycles are a result of the dynamic pricing mechanism with only two prices (high and low), instead of the strategic behavior of drivers.

Our work focuses on collective strategic behavior by drivers in response to the dynamic, state-dependent pricing. To the best of our knowledge, no existing work in the literature has studied the collective strategic behavior of groups of drivers in such a setting.  Work to this point on strategic behavior of drivers has focused on individual drivers. For example, \citet{cook2018gender} discuss learning-by-doing and the gender earnings gap; \citet{ma2019spatio} propose origin-destination based pricing that is welfare optimal and incentive compatible in the presence of spatial imbalance and temporal variation of supply and demand; \citet{garg2019driver} show that additive instead of multiplicative ``surge'' pricing is more incentive aligned for drivers when prices need to be origin-based only; \citet{rheingans2019ridesharing} study pricing in the presence of driver location preferences; and \citet{castro2021randomized} demonstrate how to use drivers' waiting times to align incentives and reduce inequity in earnings when some trips are necessarily more lucrative than the others due to operational constraints.

\section{Model \& Preliminaries} \label{sec:preliminaries}

We study a model of the interaction of drivers and riders with a ridesharing platform, and focus on the pricing and matching for trips originating from one specific region (neighborhood), such as a city center. We consider continuous time and assume that rider demand and driver supply are stationary and non-atomic. More specifically, drivers arrive at random locations in the region (e.g. where they completed their previous trip) at a rate of $\lambda_d > 0$, joining existing drivers in the region upon arrival. 
Riders arrive in the region with a rate of $\lambda_r > 0$, each looking for a trip out of the region originating from a random location in the region. This focus on a single region models situations where price cycles tend to emerge in ridesharing platforms, e.g., when a majority of riders are leaving restaurants and bars in the city center at the end of the evening.

At time $t$, the platform sets a price $p = p(t)$ for all riders requesting a trip at time $t$ and all drivers who are dispatched at time $t$. Upon the arrival of a rider, the platform dispatches the rider's trip to the closest available driver in the region, if the rider accepts the current price.
After a rider trip is dispatched to and accepted by a driver, it takes some time for the driver to pick up the rider. We denote this \newterm{en route time} (or pick-up time) as $\eta(n)$, where $n=n(t)$ is the number of available drivers in the region at time $t$. When there is a higher density of available drivers in the region, it is more likely that a rider can be matched with a nearby driver, hence the en route time is lower on average. We adopt the form $\eta(n) = \tau n^{-\alpha}$, where $\tau > 0$ can be interpreted as the average en route time when there is a single available driver in the region, and $\alpha$ is usually between $\nicefrac{1}{3}$ and $\nicefrac{1}{2}$ in practice~\citep{yan2020dynamic}. In theory, if the dispatch is to the closest drivers, then we get the above expression of $\eta$ with $\alpha = \nicefrac{1}{2}$ when demand and supply are both uniformly distributed in space~\citep{besbes2018spatial}.

Riders' value for a trip is a random variable $V$ with cumulative distribution function (CDF) $F(v)$.\footnote{A rider's value for a trip represents the rider's willingness to pay for the trip, minus the costs incurred by a driver while completing the trip, e.g., the costs of time, fuel, wear and tear, etc. Accordingly, the trip price $p$ in this paper is the payment made by the riders minus a ``base payment'' that is paid directly to the drivers to cover their costs.}
After being matched with a driver, a rider incurs a \newterm{waiting cost} of $c_r \geq 0$ per unit of time while waiting to be picked up by the driver. As a result, a rider with realized value $V = v$ has an expected utility of $v - p - c_r \eta(n)$ from a trip, when there are $n$ available drivers and her price for the trip is $p$. A rider therefore requests a trip if and only if $v \geq p + c_r \eta(n)$. We refer to $p + c_r \eta(n)$ as the \newterm{effective cost} of a trip for the riders. With $n$ available drivers, the \newterm{total rider demand per unit of time} at price $p$ is $\lambda_r \Pr[V \geq p + c_r \eta(n)] = \lambda_r \bar F(p + c_r \eta(n))$, where $\bar F(v) \triangleq 1 - F(v)$.

\newterm{Drivers' opportunity cost} is $c_d \geq 0$ per unit of time, modeling the value of their forgone outside option (e.g. driving elsewhere in the city for the same platform) while waiting for a trip dispatch in this region or driving to pick up a rider. %
We assume that there is no heterogeneity in drivers' continuation earnings after they have completed a rider trip and that drivers do not have preference over trip destinations. The expected \newterm{net payoff} of a driver after accepting a trip at price $p$ is therefore $p - c_d \eta(n)$, if there are $n$ available drivers in the region at the time of dispatch.%
\footnote{We assume that drivers' decisions on whether to accept trips are based on the \emph{expected} en route time $\eta(n)$, implicitly assuming that drivers are unable to cherry-pick rider trips based on their distance to the pick-up locations. This is aligned with practice since drivers are often penalized for declining trip dispatches~\citep{lyft2021acceptance,uberAutoOffline}. See Section~\ref{sec:avoiding_cycles} for a discussion on the impact of allowing drivers to freely choose trips they would like to accept.
}
At any point in time, an available driver may decide to remain online and accept trip dispatches, or temporarily go offline and wait for a future time to come online again. A strategic driver will act to optimize her utility (i.e., total payoff) from the trip she accepts from the current region, which is equal to the net payoff from the trip minus %
the cost she incurred waiting for this trip dispatch.

We consider a setting where the platform has full information about the en route time, riders' and drivers' arrival rates and waiting costs, the distribution of the rider values, and the present number of online drivers. However, the platform does not price trips based on the number of drivers that are offline waiting to return at a later moment. Drivers have the same information about market conditions in the region, and this is common knowledge among the drivers. Additionally, drivers are also aware of the number of offline drivers, reflecting the fact that the offline-online behavior is typically coordinated among the drivers when cycles emerge. 

For the readers' convenience, we provide a table of notations in \Cref{sec:notations}.

\subsection{System Dynamics and Steady States}

With the above model, we can now characterize the evolution of the number of available drivers in the region, $n(t) \ge 0$, over time $t \in \R$. 
Assume that drivers are non-strategic, %
stay online after arrival and do not leave the region (we model drivers' strategic behavior in \Cref{sec:defn_cycles}). Given price $p(t)$, the total rider demand is $\lambda_r \bar F\big(p(t) + c_r \eta(n(t))\big)$ per unit of time, thus the derivative of the number of drivers can be written as
\begin{equation}
  \dot n(t) = \lambda_d - \lambda_r \bar F \big(p(t) + c_r \eta(n(t))\big). \label{eq:dn_dt}
\end{equation}
The system state is characterized by the pair $(p(t),~n(t))$, where $p(t)$ is determined by the design of the platform. The goal in designing $p(t)$ is to drive the system to a steady state that optimizes the social welfare, i.e. the total rider utility minus driver costs.

\begin{definition}[Socially Optimal Steady State (SOSS)]
  A system state $(p, n)$ is said to be steady if the number of drivers is not changing, i.e., $\lambda_d - \lambda_r \bar F(p + c_r \eta(n)) = 0$. The socially optimal steady state is the steady state that maximizes social welfare per unit of time: %
  \begin{equation} 
    w \triangleq \lambda_r \bar F(p + c_r \eta(n)) \big(\E[V][V \geq p + c_r \eta(n)] - (c_r + c_d) \eta(n) \big) - c_d n. \label{eq:social_welfare}
  \end{equation}
\end{definition}

We can obtain intuition for \eqref{eq:social_welfare} by noting that $\lambda_r \bar F(p + c_r \eta(n))$ is the rider demand per unit of time (i.e., the matching rate);\footnote{We assume that the matching rate is zero if there is no available driver in the region, regardless of riders' cost of time $c_r$. %
Technically, this effectively assumes that $c_r \eta(n(t)) = \infty$ when $c_r = 0$ and $n(t) = 0$ (%
in which case $\eta(n(t)) = \infty$).
}
$\E[V][V \geq p + c_r \eta(n)]$ is the expected value of a trip for riders who request a trip at price $p$ when there are $n$ available drivers; $(c_r + c_d) \eta(n)$ is the total cost incurred by the rider and the driver due to the en route time of a trip; and $c_d n$ is the cost of the available drivers waiting to be dispatched. %

A closed form characterization of the SOSS is provided as follows.

\begin{proposition}\label{prop:SOSS}
The socially optimal steady state $(p^*, n^*)$ is given by
  \begin{align}
    n^* & = \left(\frac{\alpha \tau \lambda_d (c_r + c_d)}{c_d}\right)^{\frac{1}{\alpha + 1}}, \label{eq:n^*} \\
    p^* & = \bar F^{-1} \left(\frac{\lambda_d}{\lambda_r}\right) - c_r \eta(n^*). \label{eq:p^*}
  \end{align}
\end{proposition}

\begin{proof}
At any steady state, $\lambda_d = \lambda_r \bar F(p + c_r \eta(n))$ and $p + c_r \eta(n) = \bar F^{-1} (\lambda_d / \lambda_r)$, i.e., \eqref{eq:p^*}. The social welfare can therefore be written as:
\[w = \lambda_d \left(\E[V][V \geq \bar F^{-1} (\lambda_d / \lambda_r)] - (c_r + c_d) \tau n^{-\alpha}\right) - c_d n.\]
The highest welfare is achieved when
\[0 = \frac{\dd w}{\dd n} = \alpha \lambda_d (c_r + c_d) \tau n^{-\alpha - 1} - c_d,\]
which implies \eqref{eq:n^*}.
\end{proof}

Note that the optimal number of available drivers $n^\ast$ given by \eqref{eq:n^*} is independent of the distribution of riders' values. Intuitively, $n^\ast$ achieves an optimal trade-off between $(c_r + c_d) \eta(n)$, the total cost incurred due to the en route time, and $c_d n$, the cost of the available drivers waiting to be dispatched. Moreover, reorganizing \eqref{eq:n^*}, we have
\begin{equation}
  \frac{n^*/\lambda_d}{\eta(n^*)}
  = \frac{n^*/ \lambda_d}{\tau (n^*)^{-\alpha}}
  = \frac{\alpha  (c_r + c_d)}{c_d}. \label{equ:opt_eta_etr_ratio}
\end{equation}
In steady states, $\lambda_d$ represents both the arrival rate of drivers and the rate at which available drivers are dispatched. By Little's Law, $n^*/\lambda_d$ is the average time an available driver needs to wait to for a dispatch. \eqref{equ:opt_eta_etr_ratio} therefore implies that under the SOSS, the ratio between drivers' waiting time for a dispatch and riders' waiting time to be picked-up is fixed, independent to riders' value distribution or the arrival rates of riders and drivers.

\begin{example} \label{exmp:cycles}
Consider a %
region where drivers arrive at a rate of $\lambda_d = 1$, riders arrive at arrive at a rate of $\lambda_r = 5$, and the en route time is parameterized by $\alpha = 1 / 2$ and $\tau = 3 \sqrt{3} / 4$.
We study a few settings, where drivers' opportunity cost is $c_d = 0.03$ per unit of time, and riders' %
cost of time is either $c_r = 0$ or $c_r = 0.5$. 
Moreover, riders' values are either uniformly distributed ($V \sim \Unif[0, 1]$) or exponentially distributed ($V \sim \Exp(2.5)$).

\Cref{fig:soc_opt} shows contour lines of $\dot n(t)$, the set of steady states, and the socially optimal steady state (SOSS). 
On each contour line, we have the same value of $\dot n(t)$, and thus the same effective cost for riders $p + c_r \eta(n)$. In other words, they are the \emph{riders'} indifference curves.

\begin{figure}[t!]
  \centering
  \subcaptionbox{$c_r = 0$, $V \sim \Unif[0, 1]$. \label{fig:soc_opt_a}} {\includegraphics[width = 0.325 \textwidth] {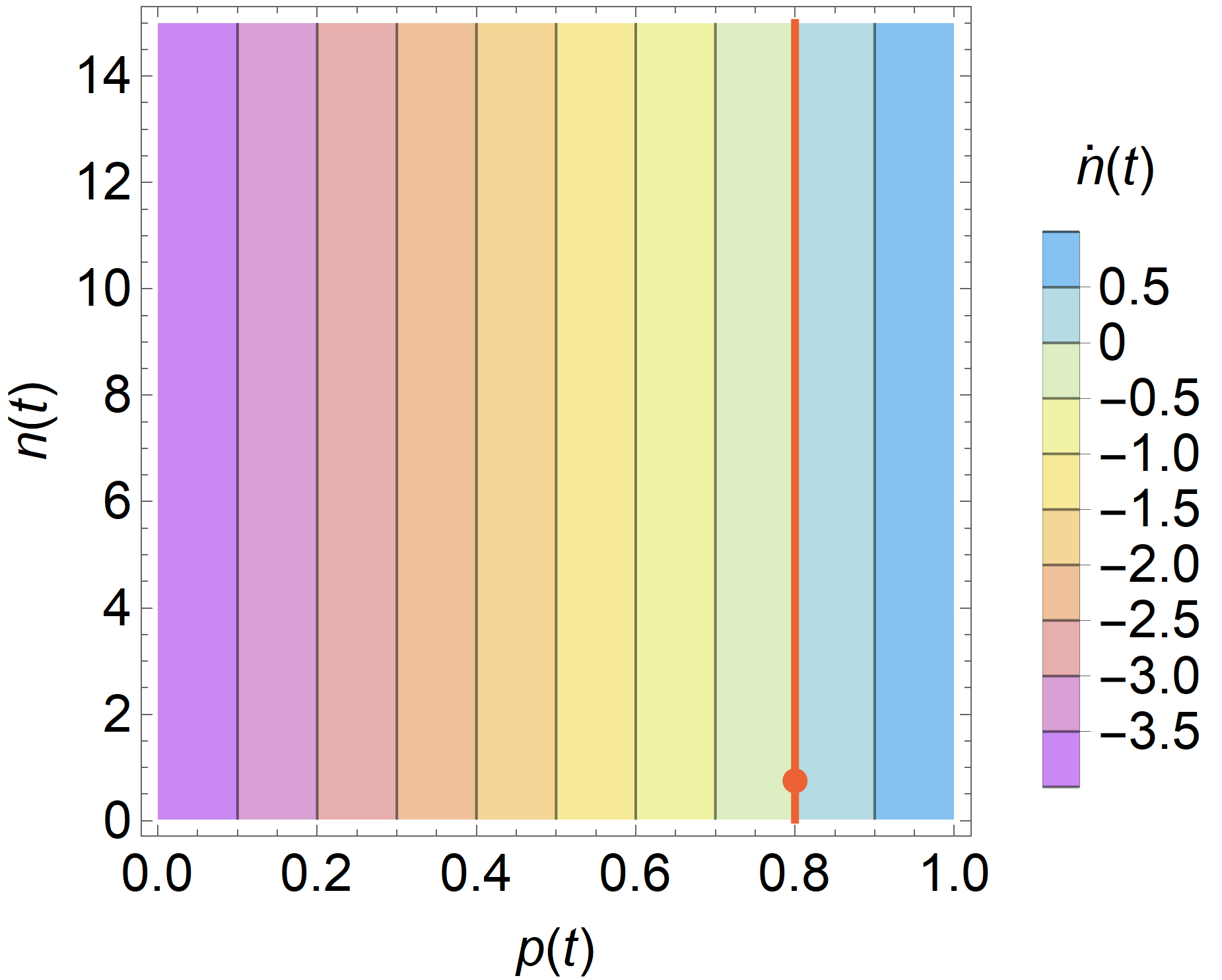}} %
  \hfill
  \subcaptionbox{$c_r = 0$, $V \sim \Exp(2.5)$. \label{fig:soc_opt_b}} {\includegraphics[width = 0.325 \textwidth] {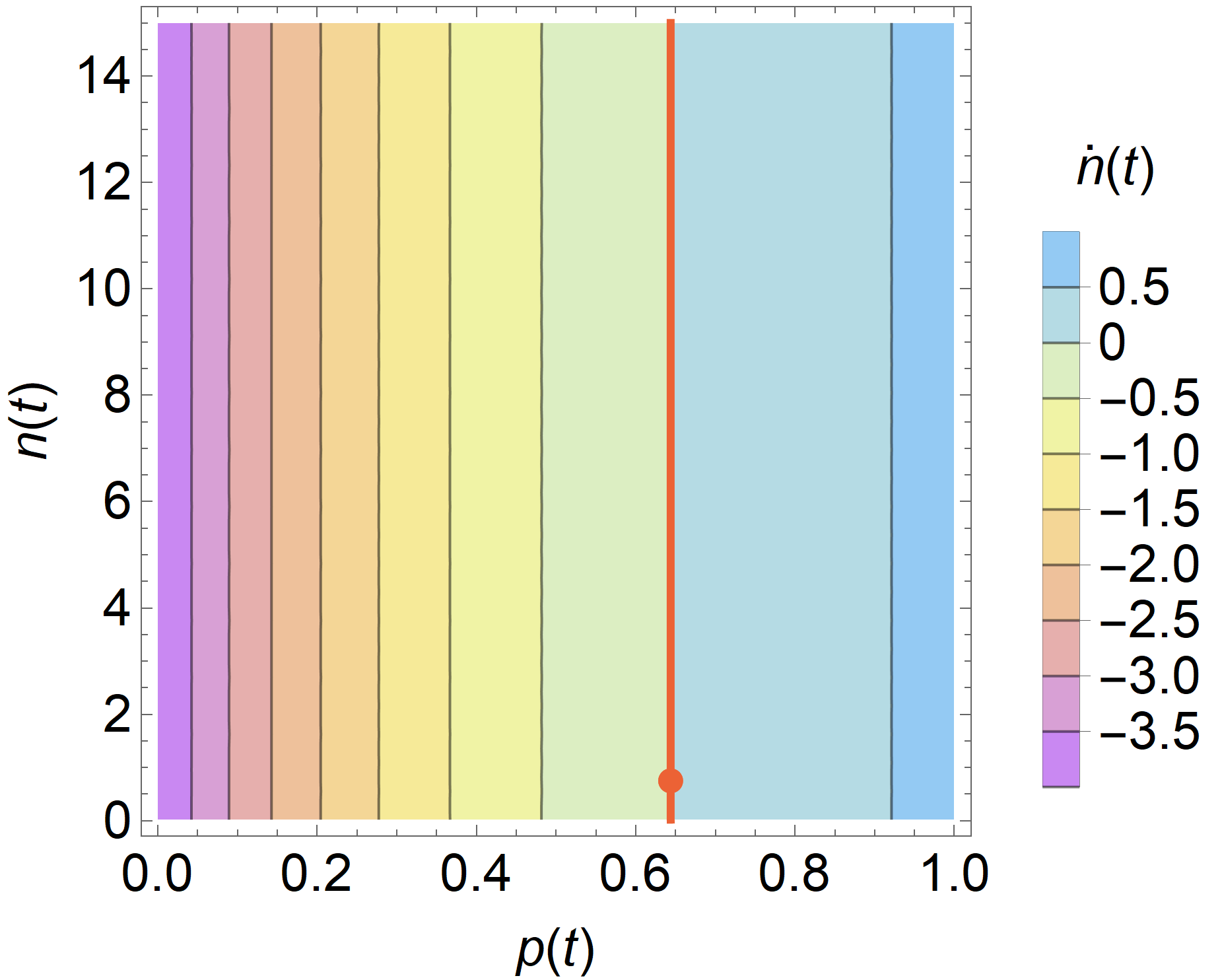}}
  \hfill
  \subcaptionbox{$c_r = 0.5$, $V \sim \Exp(2.5)$. \label{fig:soc_opt_c}} {\includegraphics[width = 0.325 \textwidth] {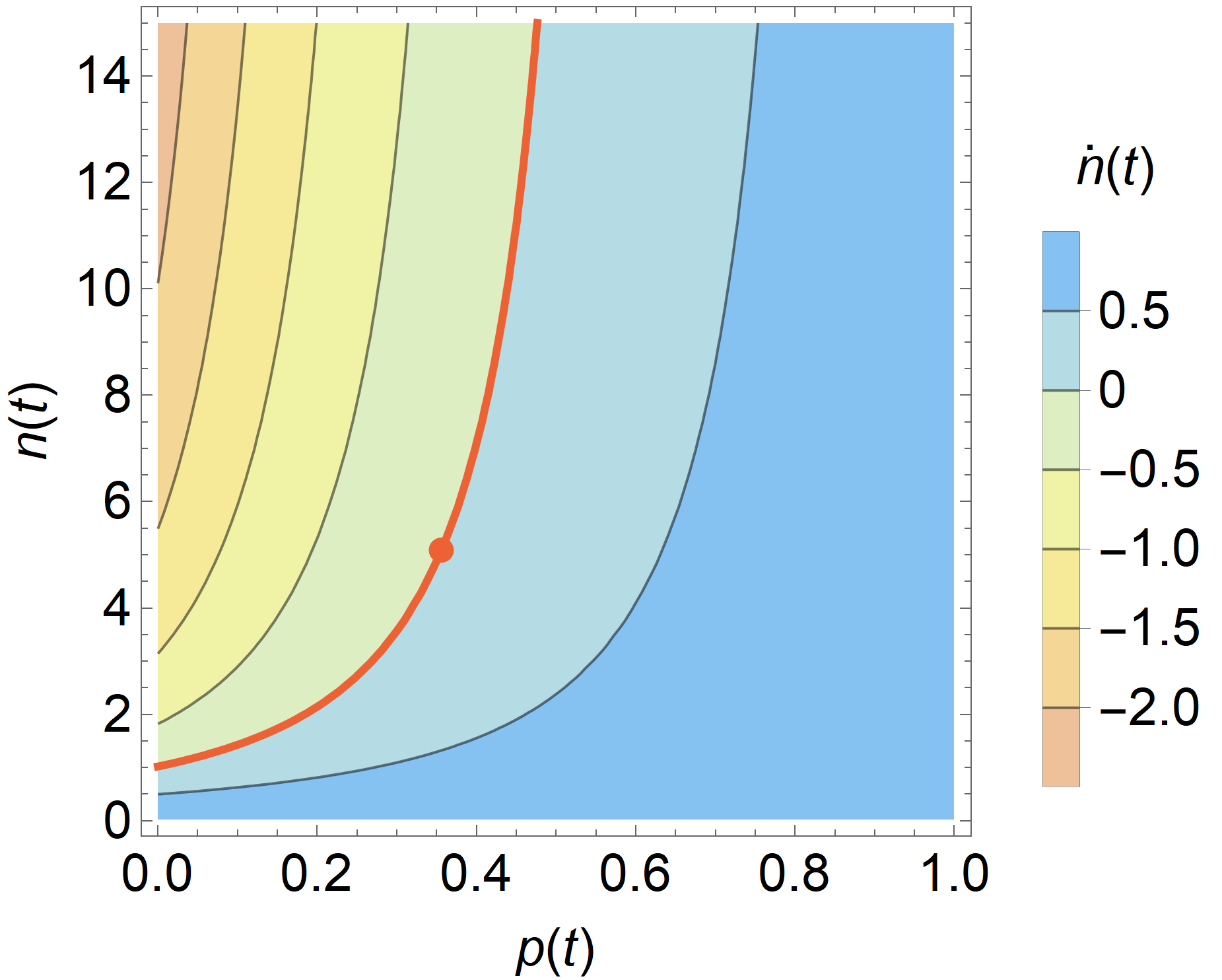}}
  \caption{Contour plot of $\dot n(t)$ at different states $(p, n)$ for \Cref{exmp:cycles}. The bold red line is the set of steady states (i.e., $\dot n(t) = 0$). The dot on the line indicates the socially optimal steady state $(p^*, n^*)$.}
  \label{fig:soc_opt}
\end{figure}

The line of steady states divides the state space into two parts: on the left hand side, the number of drivers decreases over time, whereas the number of drivers increases over time on the right hand side. 
The solid dot indicates the SOSS $(p^\ast, n^\ast)$ for each scenario. Comparing Figures~\ref{fig:soc_opt_b} and~\ref{fig:soc_opt_c}, we see that when waiting is more costly for the riders, it is optimal to maintain a higher number of available drivers %
in order to reduce the en route time. 
\end{example}

\subsection{Platform Pricing} \label{sec:pricing}

Given Proposition \ref{prop:SOSS}, we can concretely define the goal of the platform when designing the price $p(t)$. %
Specifically, we model a platform that seeks to guide the system to the socially optimal steady state %
quickly, subject to the constraint that the price cannot be adapted ``too fast.''
In practice, platforms limit the rate of price changes because 
(i) erratic prices lead to poor experiences for riders~\citep{dholakia2015everyone}; %
(ii) it takes time for drivers to re-position from one region to another, thus prices that change too fast are less effective for incentivizing driver movement; and
(iii) the underlying market condition typically changes slower than the state, i.e., the number of open cars.
In our model, we account for the need for ``smooth'' prices by limiting the rate of price increase to be at most $\ell_+ > 0$ and the rate of decrease to be at most $\ell_- > 0$, i.e., $-\ell_- \le \dot p(t) \le \ell_+$.

It remains to specify how the platform sets prices to guide the system to the SOSS quickly. Intuitively, consider an unexpected shock on the supply or the demand side that %
results in too many available drivers in the region than that under the SOSS.
The platform should first decrease the price so as to encourage higher rider demand than the arrival rate of drivers. Then, the platform should gradually increase the price back to the SOSS price. 
Similarly, if there are too few drivers available, the platform should first increase the price to reduce the rider demand, building up the pool of available drivers, %
and then decrease the price back to the SOSS price. %
Under such pricing policies, Figure~\ref{fig:pricing_policy} illustrates the how the state $(p(t), n(t))$ evolves over time for settings studied in \Cref{exmp:cycles} when the size of the price adjustments are limited by $\ell_- = \ell_+ = 0.1$, and Figure~\ref{fig:path_fastest} illustrates a  trajectory from an initial state %
to the SOSS.

To be precise, the state space (except the SOSS point) can be divided into two regions corresponding to whether the platform should increase or decrease the prices. We define the boundaries $C_+, C_- \subseteq \R^2$ as parametric curves. Intuitively, starting from any state $(p, n) \in C_+$, the trajectory eventually converges to the SOSS $(p^*, n^*)$ if the platform continues to increase the price at the highest rate $\ell_+$. Similarly, $C_-$ is the set of states leading to $(p^*, n^*)$ if the platform decreases the price at a rate of $\ell_-$. Formally, the boundaries $C_+, C_- \subseteq \R^2$ are defined as follows.
\begin{itemize}[leftmargin=*]
  \item Let $(p_+(t), n_+(t))$ be the solution to the following system of differential equations: 
  \[\dot p(t) = \ell_+, \qquad \dot n(t) = \lambda_d - \lambda_r \bar F \big(p(t) + c_r \eta(n(t))\big), \qquad p(0) = p^*, \qquad n(0) = n^*.\]
  Define parametric curve $C_+ \triangleq \Set{(p_+(t),~ n_+(t)) \in \R^2 \Given t \in (-p^* / \ell_+, 0)}$.

  \item Let $(p_-(t), n_-(t))$ be the solution to the following system of differential equations:
  \[\dot p(t) = -\ell_-, \qquad \dot n(t) = \lambda_d - \lambda_r \bar F \big(p(t) + c_r \eta(n(t))\big), \qquad p(0) = p^*, \qquad n(0) = n^*.\]
  Define parametric curve $C_- \triangleq \Set{(p_-(t),~ n_-(t)) \in \R^2 \Given t \in (-\infty, 0)}$.
\end{itemize}

Given the definitions above, $C_+ \cup C_- \cup \{(p^*, n^*)\}$ divide the state space into two regions. In the upper (or lower) region, %
the platform first decreases (or increases) the price, and then starts to increase (or decrease) the price again when the system state reaches $C_+$ (or $C_-$).
 
Formally, the platform's \newterm{pricing policy} $\pi$ is defined as follows.

\begin{definition} \label{defn:pricing_policy}
The platform's pricing policy $\pi: \R^2 \to \R$ is mapping from a state $(p, n)$ to the derivative of the price $\dot p$ given by:
\begin{equation} \label{eq:pricing_policy}
  \dot p = \pi(p, n) = \begin{cases}
    0       & \text{if } (p, n) = (p^*, n^*). \\
    \ell_+  & \text{if } \exists n' \ge n, (p, n') \in C_+ \text{ or } \exists n' > n, (p, n') \in C_-. \\
    -\ell_- & \text{if } \exists n' < n, (p, n') \in C_+ \text{ or } \exists n' \le n, (p, n') \in C_-.
  \end{cases}
\end{equation}
\end{definition}

We illustrate the structure of the platform pricing policy by returning to our running example.

\begin{figure}[t!]
  \centering
  \subcaptionbox{$c_r = 0$, $V \sim \Unif[0, 1]$.} {\includegraphics[width = 0.325 \textwidth] {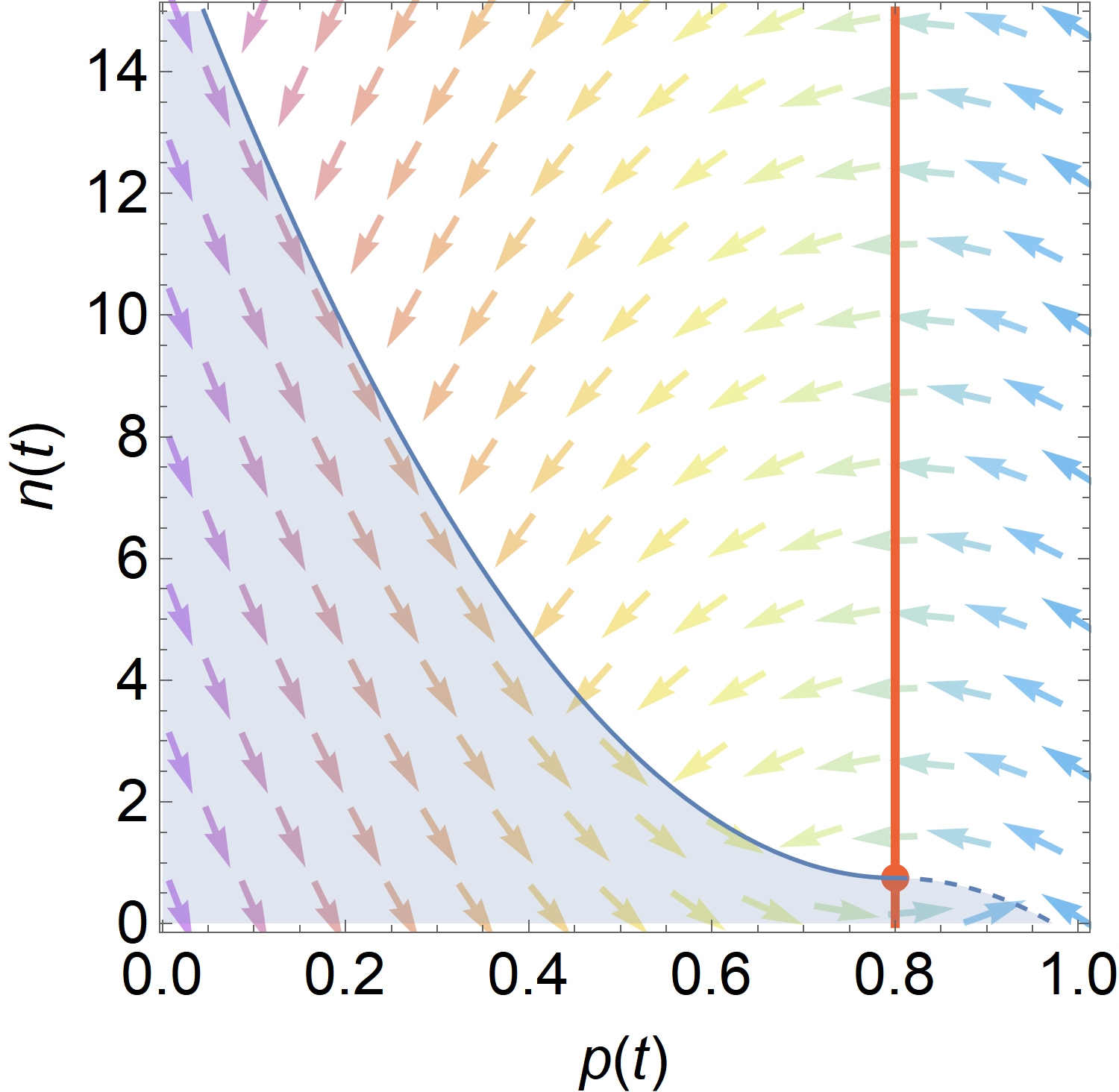}}
  \hfill
  \subcaptionbox{$c_r = 0$, $V \sim \Exp(2.5)$.} {\includegraphics[width = 0.325 \textwidth] {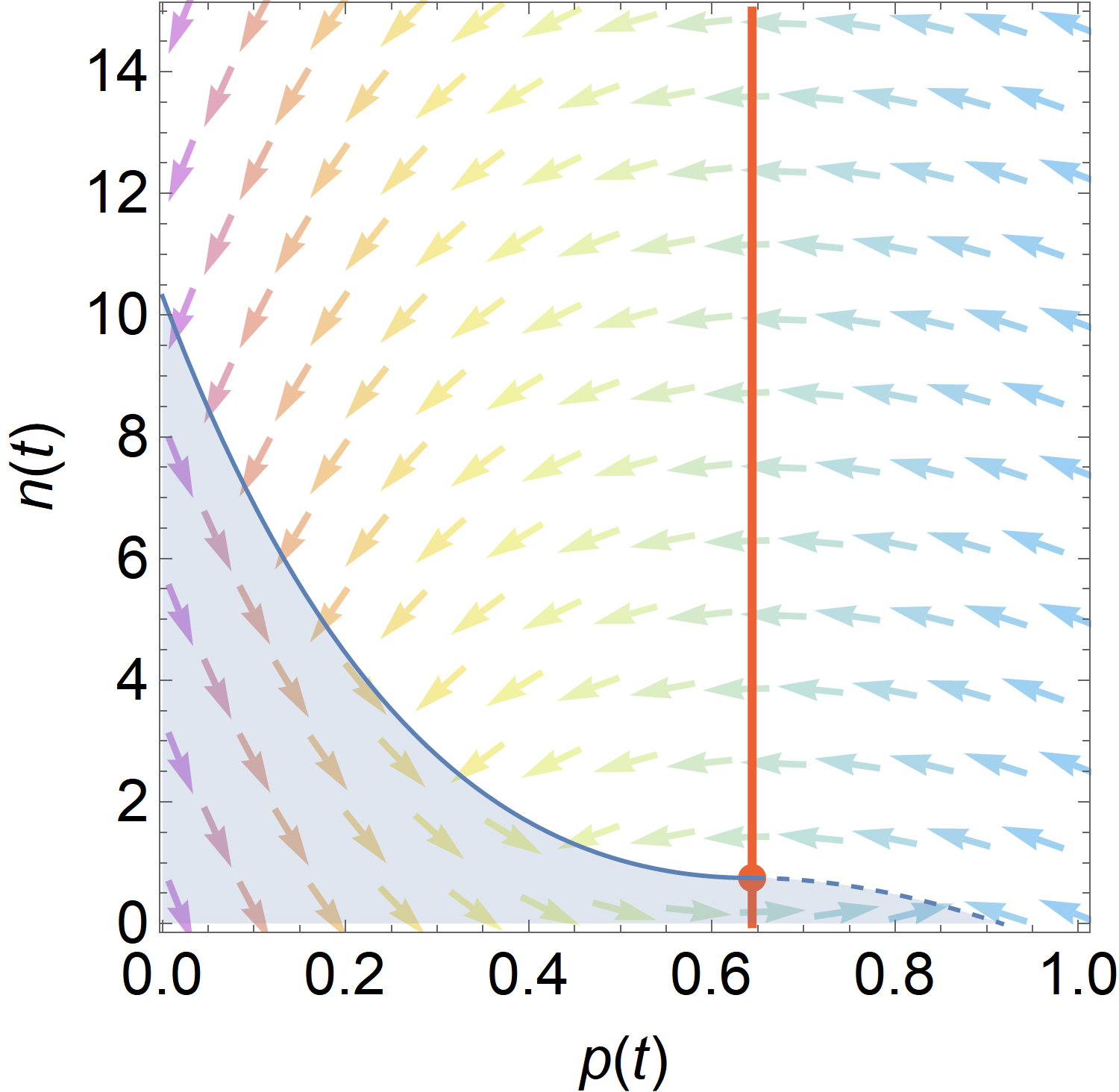}}
  \hfill
  \subcaptionbox{$c_r = 0.5$, $V \sim \Exp(2.5)$.} {\includegraphics[width = 0.325 \textwidth] {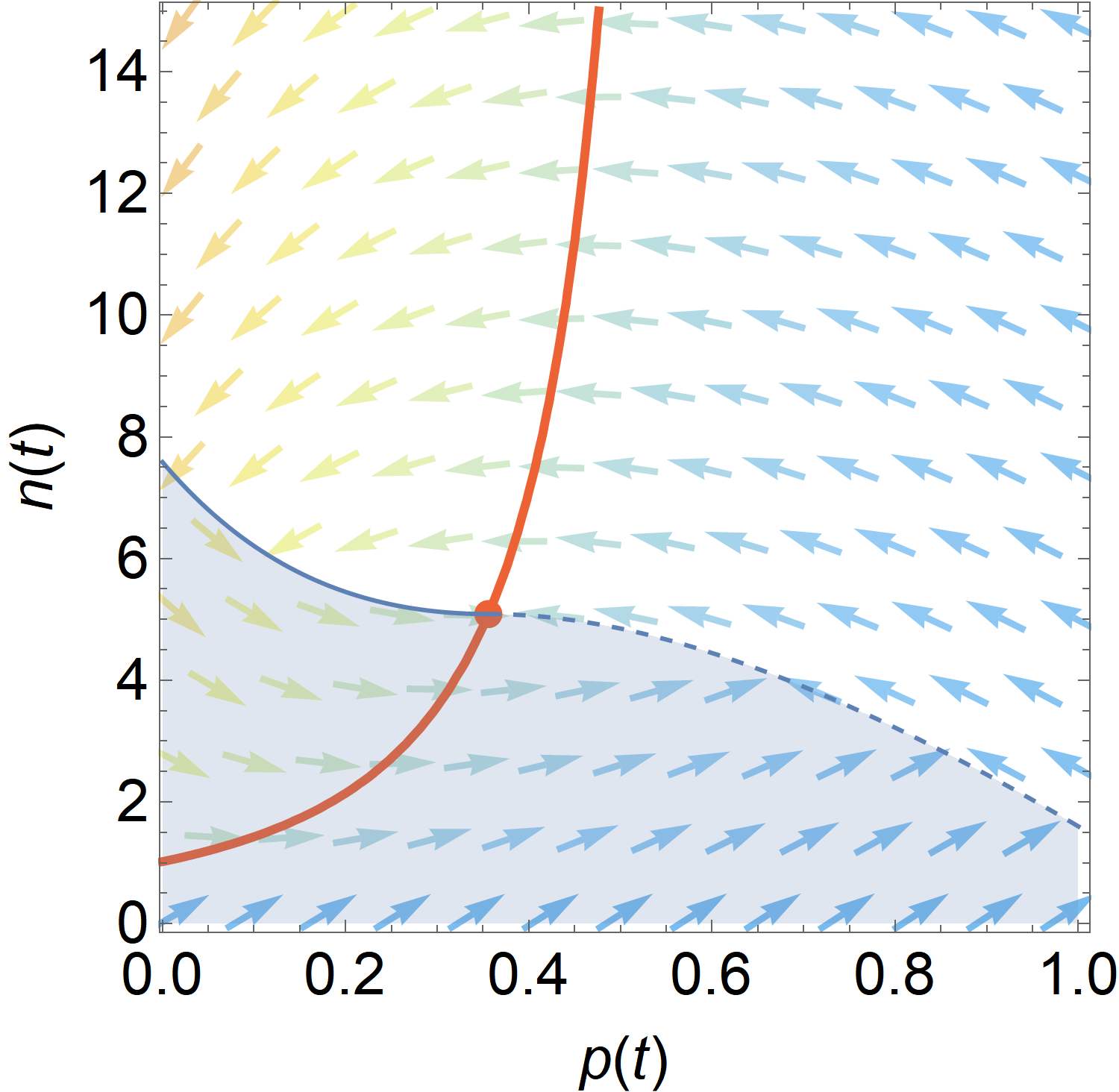}}
  \caption{Illustration of platform pricing policy and state evolution, with $\ell_- = \ell_+ = 0.1$. The vector field indicates the directions of state evolution, and the arrow colors represent how fast the number of drivers changes. In the upper region and on the dashed line $C_-$, the platform reduces the price. In the lower region and on the solid line $C_+$, the platform raises the price.
  }
  \label{fig:pricing_policy}
\end{figure}

\begin{figure}[t!]
  \vspace{.1in}
  \centering
  \subcaptionbox{Path under the platform pricing policy. Time $= 1 + 6 \sqrt{3} \approx 11.4$.\label{fig:path_fastest}} {\includegraphics[height = 0.293 \textwidth] {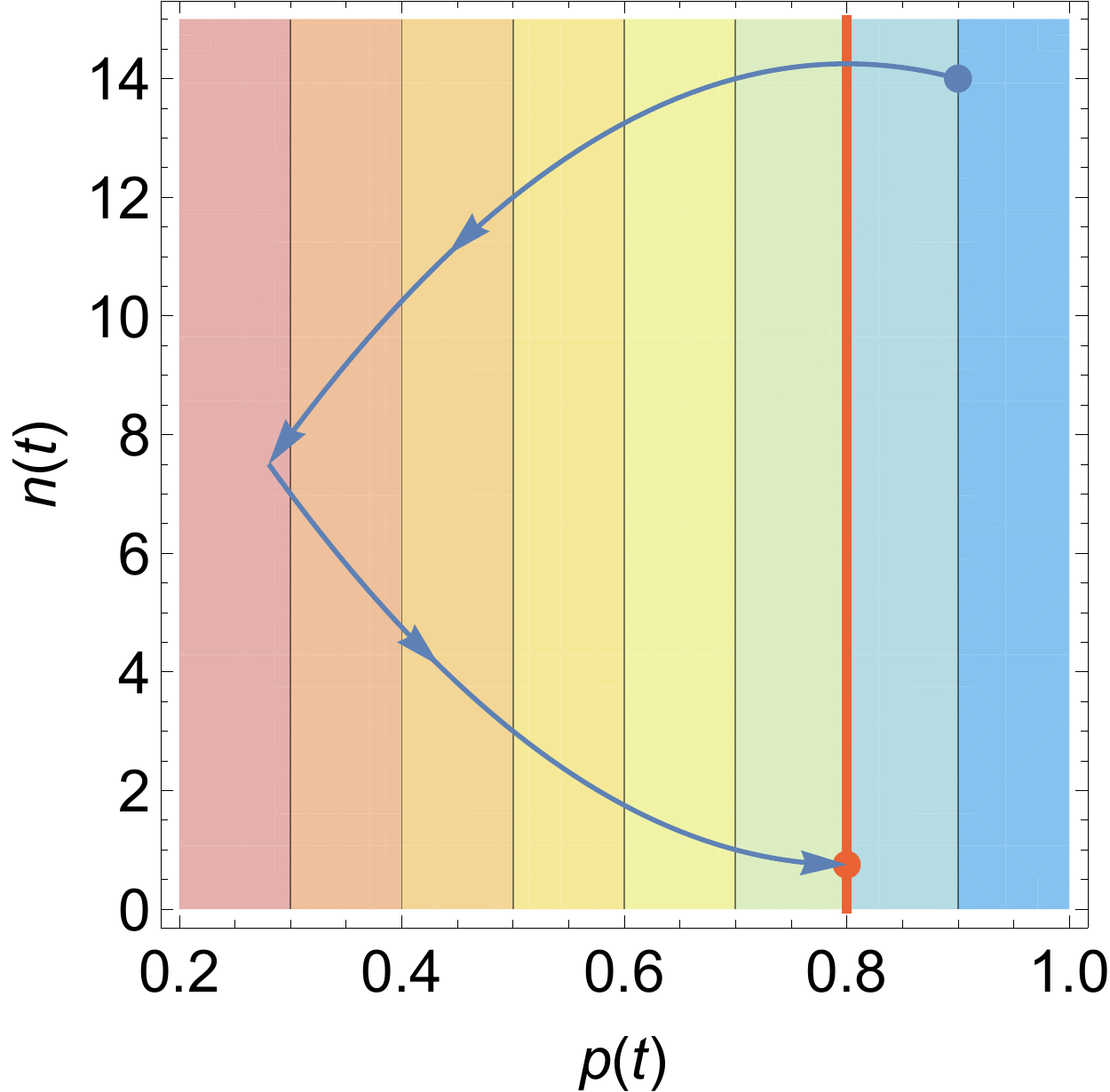}}
  \hfill
  \subcaptionbox{Path that undershoots. \\Time $= 13$. \label{fig:path_undershoots}} {\includegraphics[height = 0.293 \textwidth] {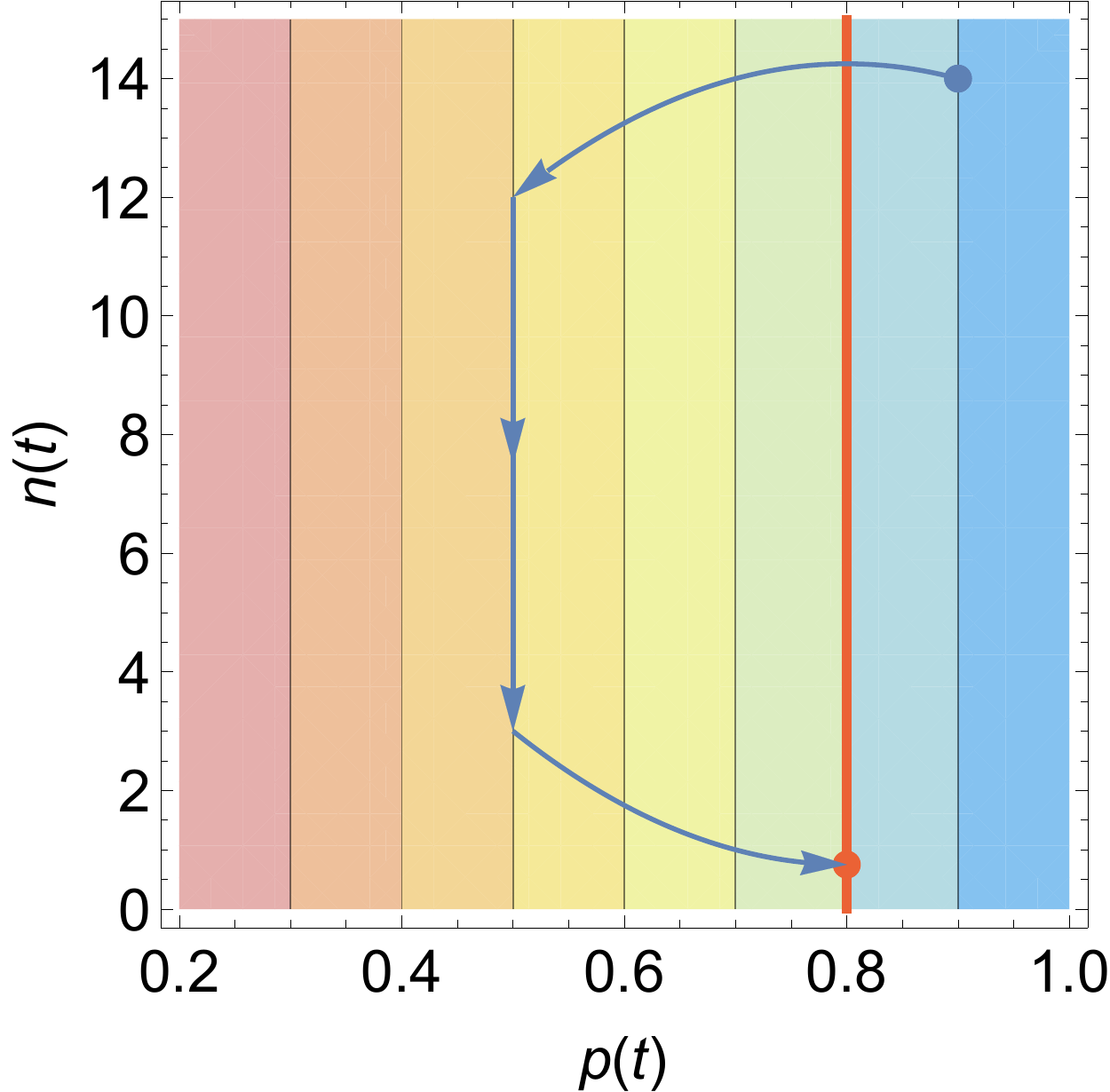}}
  \hfill
  \subcaptionbox{Path that overshoots. \\Time $= 3 + 4 \sqrt{7} \approx 13.6$. \label{fig:path_overshoots}}
  {\includegraphics[height = 0.293 \textwidth] {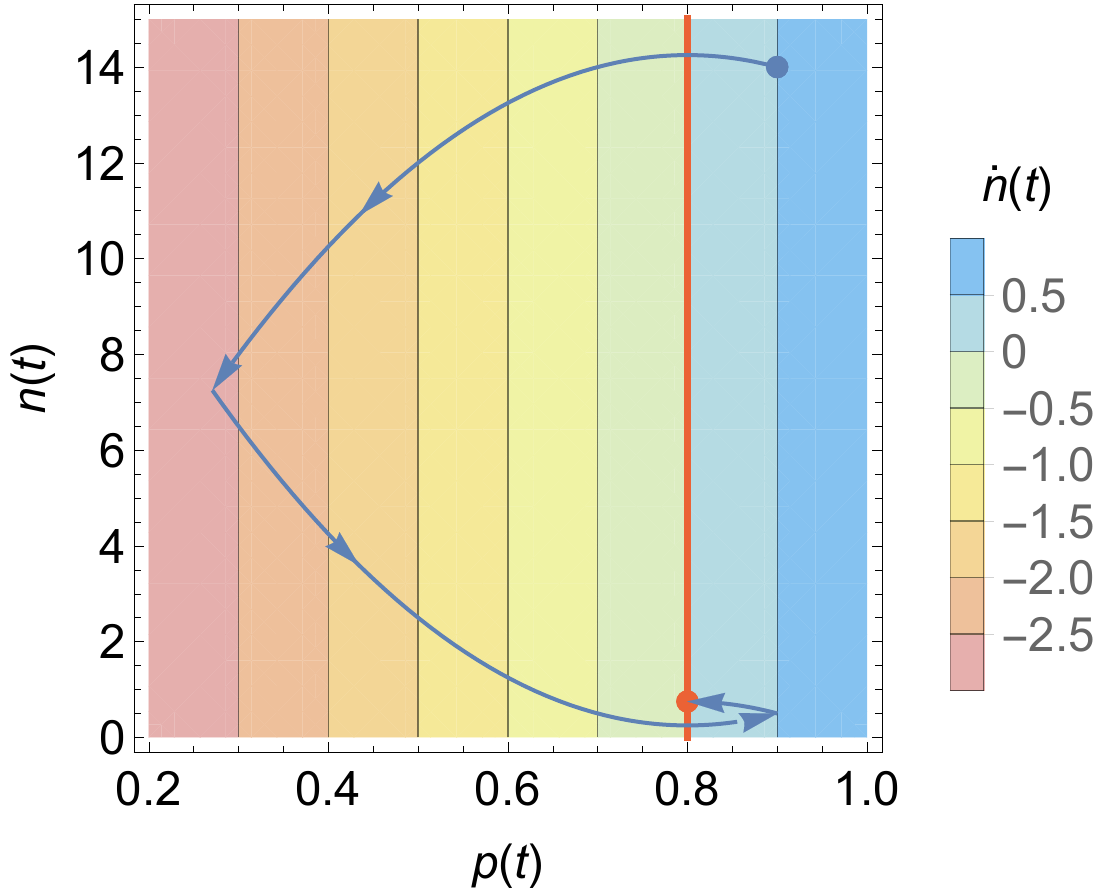}}
  \caption{Three paths from state $(0.9, 14)$ to the SOSS $(0.8, 0.75)$ when $c_r = 0$,  $V \sim \Unif[0, 1]$, and $\ell_- = \ell_+ = 0.1$. %
  }
  \label{fig:under/over-shoot}
\end{figure}

\begin{examplecont}{exmp:cycles}
Suppose the rates for price changes are limited by $\ell_+ = \ell_- = 0.1$. \Cref{fig:pricing_policy} shows the pricing policy of the platform for the three scenarios analyzed in \Cref{fig:soc_opt}, including a vector field indicating how the state evolves over time.

For the scenario with rider cost $c_r = 0$ and rider value $V \sim \Unif[0, 1]$, Figure~\ref{fig:under/over-shoot} provide three example paths to move from an initial state $(p,n) = (0.9, 14)$ to the SOSS $(p^\ast, n^\ast) = (0.8, 0.75)$.
Figure~\ref{fig:path_fastest} illustrates the trajectory under the platform's pricing policy given in \eqref{eq:pricing_policy}. 
In Figure~\ref{fig:path_undershoots}, prices was held at 0.5 without decreasing further, leading to a slower return to the SOSS. On the other hand, when prices continued to decrease after reaching $C_+$ in Figure~\ref{fig:path_overshoots}, there are less driver than $n^\ast$ when the price rose to $p^\ast$, and this again leads to a slower return to the SOSS.
\end{examplecont}

\section{Stable Price Cycles} \label{sec:defn_cycles}

The focus of this paper is on understanding and controlling price cycles in ridesharing platforms. In this section, we demonstrate how price cycles emerge when drivers collectively adopt \emph{online/offline strategies}, going offline at the same time, and coming back online to take rides again only after prices have risen due to the seemingly low supply level.

Before formally defining an online/offline strategy, we first illustrate why it can be beneficial for a driver to go offline and temporarily and withhold supply. 
Consider, for example, the path from the initial state $(0.9, 14)$ back to the SOSS as shown in Figure~\ref{fig:path_fastest}. Figure~\ref{fig:offline_motivation} illustrates the net payoff a driver gets from a accepting trip $p(t) - c_d \eta(n(t))$, assuming that all drivers are non-strategic, and that the state is at $(0.9, 14)$ at time $t = 0$. Intuitively, drivers get a low net payoff from accepting trips during times with low prices. 

\begin{figure}[t!]
  \centering
  \includegraphics[width = 0.6 \textwidth] {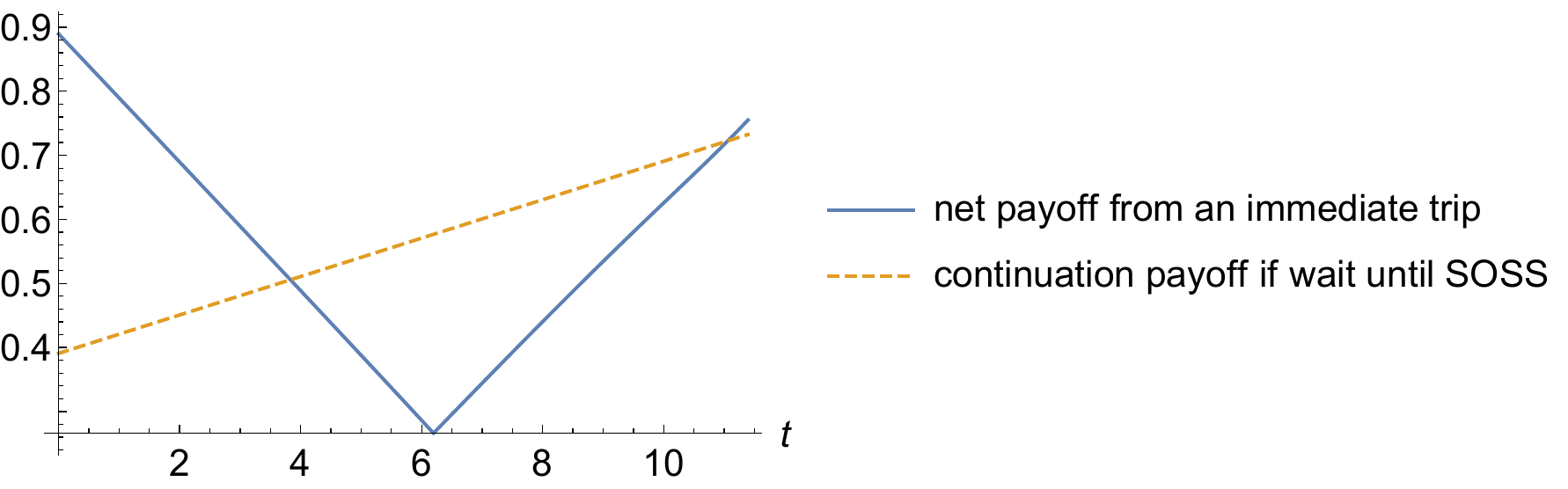}
  \caption{Along the path in \Cref{fig:path_fastest}, the net payoff from an immediate trip $p(t) - c_d \eta(n(t))$, and the continuation payoff of remaining offline until the system state reaches the SOSS $u^* - c_d (%
  11.4 - t)$. %
  }
  \label{fig:offline_motivation}
\end{figure}

Instead of staying online and accepting the first dispatch from the platform, consider a driver who stays offline until the state reaches the SOSS. 
At the SOSS, each driver waits an average of $n^*/\lambda_d$ units of time for a trip dispatch. The \newterm{continuation payoff} of a driver at the SOSS, i.e. the net payoff from the trip minus the cost the driver incurred waiting for the dispatch, is therefore: 
\begin{equation}
  u^* \triangleq p^* - c_d \eta(n^*) - c_d \frac{n^*}{\lambda_d}.
\end{equation}
Since it takes $1 + 6 \sqrt{3} \approx 11.4$ units of time for system to return to the SOSS, at time $t$, the continuation payoff of a driver who stays offline until the SOSS is approximately $u^\ast - c_d(11.4 - t)$. This is also illustrated in Figure~\ref{fig:offline_motivation}. 
As we can see, if the rest of drivers stay online and accept all dispatches from the platform, a driver can benefit from declining trip dispatches during times when prices are low (from $t \approx 4$ to $t \approx 11$, in this example). This illustrates that pricing trips in order to quickly return to the SOSS may incentivize drivers to strategically go offline and wait for a better price.%
\footnote{Note that ``remaining offline until the SOSS'' is not the optimal strategy a driver can adopt in this scenario. See Section~\ref{sec:cycle_existence} for a formal discussion on drivers' best response.}

\subsection{Online/Offline Strategies}

In Section~\ref{sec:preliminaries}, we define $n(t)$ as the \newterm{number of {online} drivers} on the platform, i.e., those who are in the region and available for dispatch at time $t$. Here, we additionally denote $n_0(t) \ge 0$ as the \newterm{number of {offline} drivers} in the region at time $t$, i.e., those waiting to return at a later moment. Let $N(t) \triangleq n(t) + n_0(t)$ be the \newterm{total number of drivers} at time $t$.
The platform sets trip prices based on $n(t)$ but not $n_0(t)$. A \newterm{driver strategy} $\sigma$ determines, for any time $t$, the probability $\sigma(t) \in [0, 1]$ with which the driver is online and accepts dispatches from the platform; thus drivers is offline at time $t$ with probability $1 - \sigma(t)$.

In practice, drivers coordinate on a kind of simple ``threshold strategies''~\citep{WJLA2019ReganNationalCollusion,uberPeopleLA}, going offline when prices drop below some $\ubar p > 0$ and returning online after the price has risen above some $\bar p > \ubar p$. Such strategies depend on the price level $p$ and whether the price is increasing or decreasing, and are affected by the number of online drivers $n$ indirectly through the platform's pricing policy $\dot p = \pi(p, n)$.

\begin{definition}[Online/offline threshold strategies] \label{defn:online_offline_strategy}
Under an online/offline threshold strategy with upper and lower price thresholds $(\bar p, ~\ubar p)$, a driver switches online after the price rises above $\bar p$, and goes offline when price drops below $\ubar p$. Formally,
\begin{equation} \label{eq:driver_strategy}
  \sigma_{\bar p, \ubar p}(p, n) = \begin{cases}
  1 & \text{if } (p \ge \bar p) \text{ or } (p > \ubar p \text{ and } \pi(p, n) \le 0), \\
  0 & \text{if } (p \le \ubar p) \text{ or } (p < \bar p \text{ and } \pi(p, n) > 0), 
  \end{cases}
\end{equation}
where $\pi$ is the platform's pricing policy as in \eqref{eq:pricing_policy}.
\end{definition}

When all drivers coordinate and adopt the same online/offline threshold strategy, at any point in time $t$, either all drivers are online or all are offline.
When adopted by drivers, these kinds of threshold strategies have the potential to lead to cyclic fluctuations of prices and the number of available drivers on the platform.
Consider, for example, the scenario analyzed in Figure~\ref{fig:path_fastest}. When %
drivers employ an online/offline threshold strategy with $(\ubar p, \bar p) = (0.3, 0.9)$, %
the state evolves as follows (see Figure~\ref{fig:stable_cycle} for an illustration).
\begin{itemize}[leftmargin=*]
    \item Starting from the initial state $(0.9, 14)$, the platform decreases the price in order to expedite the dispatching of the excessive driver supply. The number of drivers gradually declines after a small initial increase. Once the price drops to $\ubar p = 0.3$, however, all drivers switch offline. With zero available drivers online, the platform starts to increase the price at a rate of $\ell_+$. %
    \item At the time the price reaches $\bar p = 0.9$, all offline drivers re-emerge and start to accept dispatches again. With the large number of drivers accumulated during the period of price increase, the state $(p, n)$ will likely reside above $C_-$. As a result, the platform will again start to decrease the trip prices at a rate of $\ell_-$, and potentially repeating the same dynamic if price again drops to $\ubar p$.
\end{itemize}

We describe such cyclic fluctuations of prices using times $t_0$, $t_1$ and $t_2$, where $t_0$ is a time when the price is the lowest (i.e. $p(t_0) = \ubar p$), $t_1$ is the first time after $t_0$ such that the price reaches the peak %
(i.e. $p(t_1) = \bar p$), and $t_2$ is the first time after $t_1$ s.t. $p(t_2) = \ubar p$, which is exactly one period after $t_0$. 
Under the platform's pricing policy \eqref{eq:pricing_policy}, we have $t_1 - t_0 = (\bar p - \ubar p)/\ell_+$ and $t_2 - t_1 = (\bar p - \ubar p)/\ell_-$, thus the \newterm{period length} is $T \triangleq t_2 - t_0 = (\bar p - \ubar p)(1/\ell_+ + 1/\ell_-)$.

Further, we use $u(t, \sigma, \sigma')$ to denote the \newterm{continuation payoff} of a driver from time $t$ onward, when the driver adopts strategy $\sigma$ while others employ strategy $\sigma'$.
It is defined as the price of the trip the driver accepts, minus the cost of time the driver incurs before picking up the rider: %
\begin{equation}
  u(t, \sigma, \sigma') \triangleq \E[p(t + T_\text{wait}) - c_d (T_\text{wait} + \eta(n(t + T_\text{wait})))]. \label{equ:expected_utility}
\end{equation}
Here, $T_\text{wait}$ is a random variable representing the amount of time a driver waits before accepting a trip dispatch, when the driver adopts strategy $\sigma$ and when every other driver employs strategy $\sigma'$. 
$p(t + T_\text{wait})$ is the trip price at the time of dispatch, and $\eta(n(t + T_\text{wait}))$ is the average en route time the driver spends picking up the rider. 

With this notation in hand, we now formally define a stable price cycle.

\begin{definition}[Stable price cycle] \label{def:stable_cycles} 
Suppose the rate of price adjustments is limited by $(\ell_+, \ell_-)$. 
An online/offline threshold strategy $\sigma_{\bar p, \ubar p}$ forms a \emph{stable price cycle} with a maximum of $\hat n > 0$ drivers if the associated trip price $p(t)$ and number of drivers online $n(t)$ satisfy $\max_{t' \in [t, t + T]} n(t') = \hat{n}$, and:
  \begin{enumerate}[label = (C\arabic*)]
    \item (market clearing) \label{cond:repeat} the cycle repeats itself, i.e., the net accumulation of drivers over a period is zero. Formally, for any time $t$,
    \begin{equation} \label{eq:repeat}
      \lambda_d T = \int_{t}^{t + T} \lambda_r \bar F(p(t') + c_r \eta(n(t'))) \dd t';
    \end{equation}
    \item (driver best response) \label{cond:driver} the online-offline strategy $\sigma_{\bar p, \ubar p}$ forms a Nash equilibrium among drivers, meaning that it is not useful for a %
    driver to unilaterally deviate to any feasible (and potentially non-threshold) strategy $\sigma$, i.e., 
    \begin{equation}
      u(t, \sigma_{\bar p, \ubar p}, \sigma_{\bar p, \ubar p}) \ge u(t, \sigma, \sigma_{\bar p, \ubar p}), \ \forall \sigma, \ \forall t.
    \end{equation}
  \end{enumerate}
\end{definition}

For simplicity of notation, we denote a stable price cycle succinctly as $(\ell_+, \ell_-, \ubar p, \bar p, \hat n)$. As an illustration, we return to the setting analyzed in \Cref{exmp:cycles}.

\begin{examplecont}{exmp:cycles}
Consider the setting where the arrival rate of drivers and riders are $\lambda_d = 1$ and $\lambda_r = 5$, respectively. The en-route time is parameterized by $\alpha = 1 / 2$ and $\tau = 3 \sqrt{3} / 4$. Riders' value is uniformly distributed $V \sim \Unif[0, 1]$, and riders' and drivers' costs of time are $c_r = 0$ and $c_d = 0.03$. 
\Cref{fig:stable_cycle_1} illustrates a stable cycle under a pricing policy with $\ell_+ = \ell_- = 0.1$, when all drivers adopt an online/offline threshold strategy $\sigma_{\bar p, \ubar p}$ with $\ubar p = 0.3$ and $\bar p = 0.9$. \Cref{fig:stable_cycle_1_state} shows how the state evolves over time; \Cref{fig:stable_cycle_1_number_of_drivers} illustrates the number of drivers who are online and offline at different points of time; \Cref{fig:stable_cycle_1_price_and_utility} presents the trip price, drivers' net payoff, and drivers' continuation payoff under strategy $\sigma_{\bar p, \ubar p}$. As a comparison, drivers' continuation payoff under the socially optimal steady state is $u^* = p^\ast - c_d(\eta(n^\ast) + n^\ast / \lambda_d)  = 0.7325$.

\begin{figure}[t!]
  \centering
  \begin{minipage}{0.41 \hsize} 
    \subcaptionbox{Driver number versus price. \label{fig:stable_cycle_1_state} }{\includegraphics[height = \hsize]{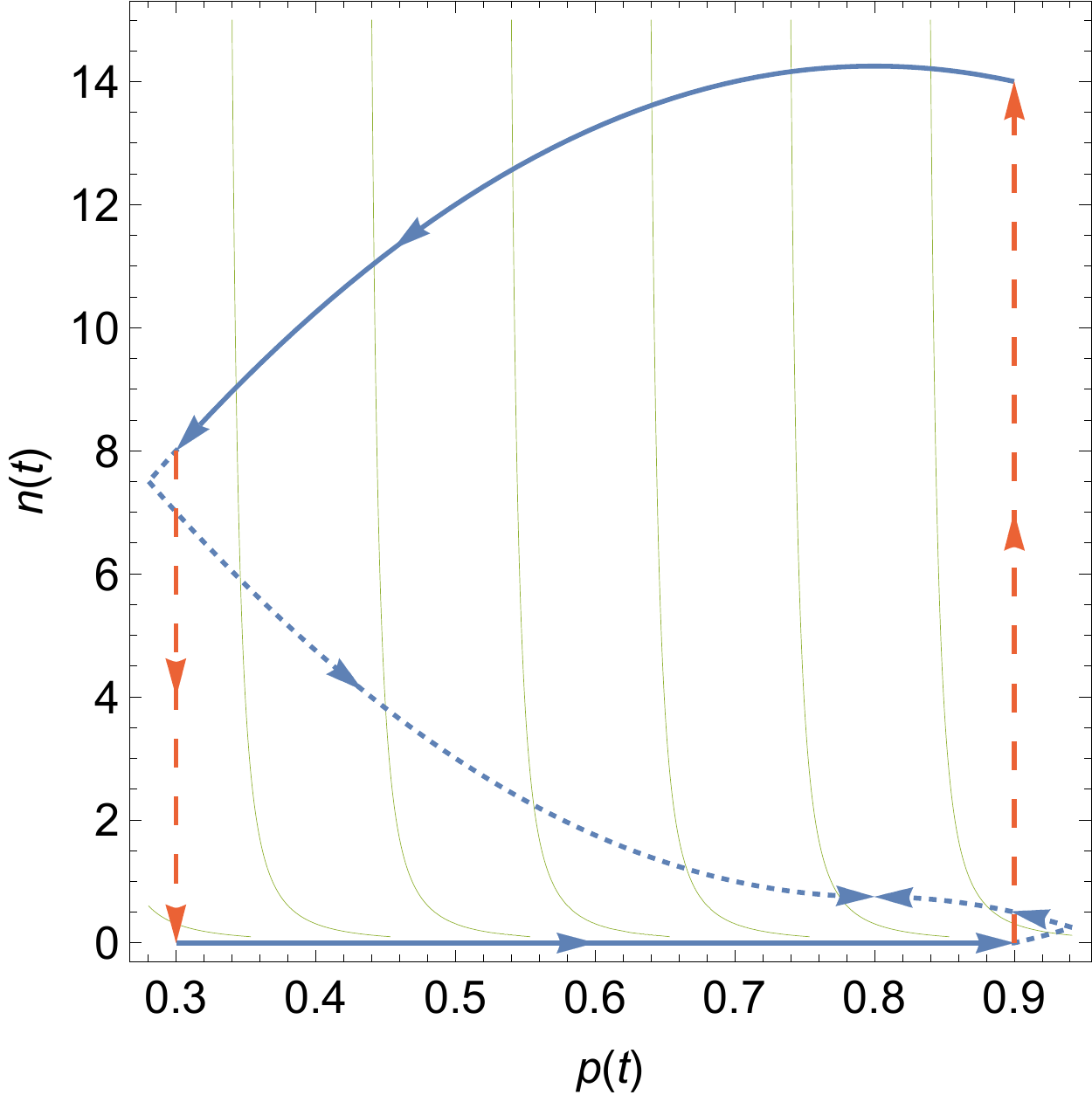}}
  \end{minipage}
  \hfill
  \begin{minipage}{0.55 \hsize}
    \:\subcaptionbox{Driver number versus time. 
    \label{fig:stable_cycle_1_number_of_drivers}}{\includegraphics[height = .34 \hsize]{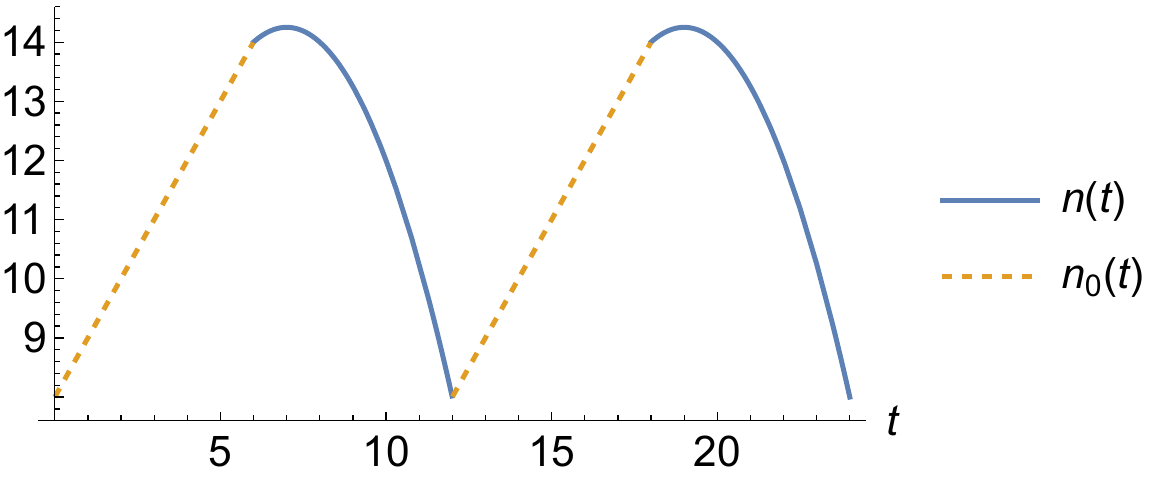}}
    \\
    \subcaptionbox{Price and utility versus time. \label{fig:stable_cycle_1_price_and_utility}}{\includegraphics[height = .33 \hsize]{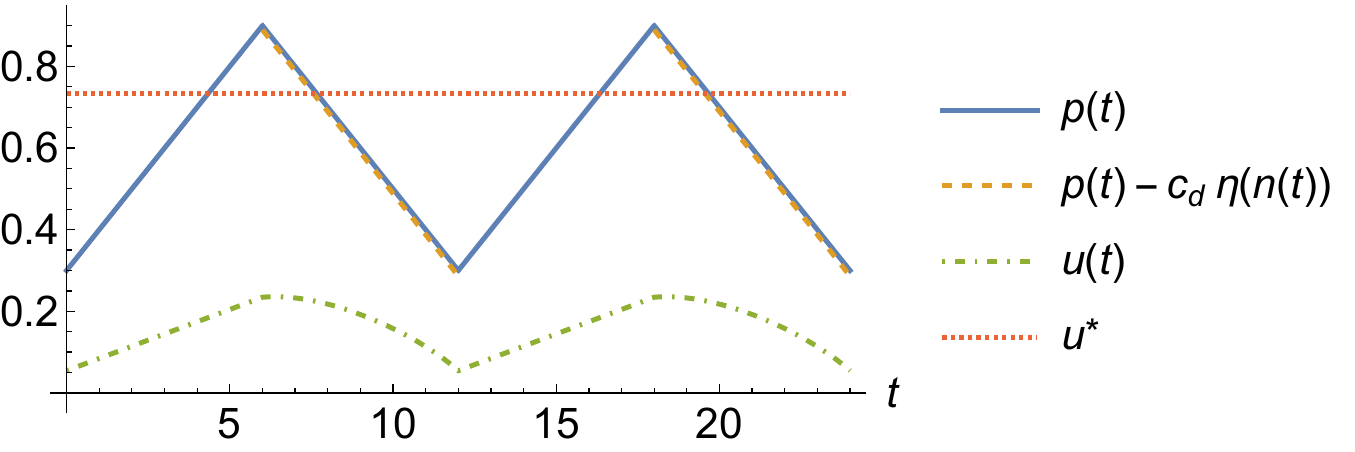}}
  \end{minipage}
  \caption{The ``symmetric'' cycle (with $\ell_+ = \ell_- = 0.1$) in \Cref{exmp:cycles}.
  \Cref{fig:stable_cycle_1_state} shows the trajectories %
  induced by the online/offline threshold strategy %
  (solid blue lines) and the trajectories planned by the platform's pricing policy (dotted blue lines). 
  The vertical red dashed lines indicates drivers' collectively going online and offline. The green lines are the indifference curves of the \emph{drivers}. 
  \Cref{fig:stable_cycle_1_number_of_drivers} shows the number of offline drivers $n_0(t)$ and online drivers $n(t)$ over time $t$.
  \Cref{fig:stable_cycle_1_price_and_utility} shows the trip price $p(t)$, a driver's net payoff $p(t) - c_d \eta(n(t))$, a driver's continuation payoff during the cycle $u(t)$, a driver's continuation payoff at SOSS $u^*$.
  Drivers go offline at time $t_0 = 0$, then switch online at time $t_1 = 6$, and go offline again at time $t_2 = 12$.
  \label{fig:stable_cycle_1}
  }
  \label{fig:stable_cycle}
\end{figure}

Under $\sigma_{\bar p, \ubar p}$, when the trip price increases from $\bar p = 0.3$ to $\ubar p = 0.9$, all drivers (including those who arrive meanwhile) remain offline. The drivers come online once the price reaches $0.9$, and the platform then starts to lower the price after observing the driver supply. There are more drivers turning online at $\bar p$ than the number of drivers going offline at $\ubar p$, because drivers continue to arrive as the price increases.

In \Cref{fig:stable_cycle_2}, we illustrate a stable cycle where the rates of price increase and decrease are asymmetric: $\ell_+ = 0.2$, and $\ell_- = 0.1$. The setting is the same as above, except that $c_d = 0.08$ and $\tau = 2$. 

\begin{figure}[t!]
  \centering
  \begin{minipage}{0.41 \hsize}
    \;\subcaptionbox{Driver number versus price. \label{fig:stable_cycle_2_state} }{\includegraphics[height = \hsize]{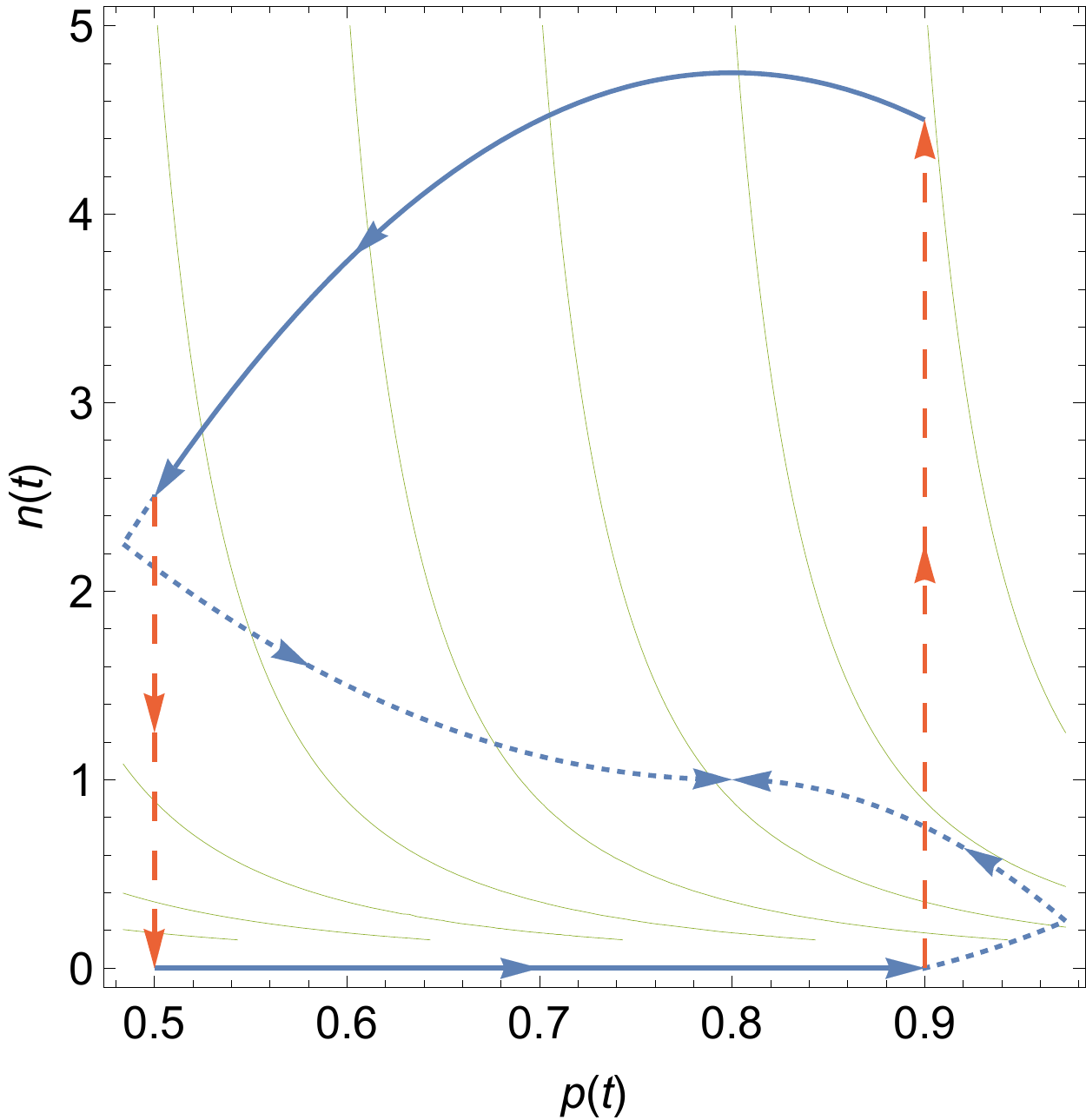}}
  \end{minipage}
  \hfill
  \begin{minipage}{0.55 \hsize}
    \subcaptionbox{Driver number versus time. %
    \label{fig:stable_cycle_2_number_of_drivers}}{\includegraphics[height = .34 \hsize]{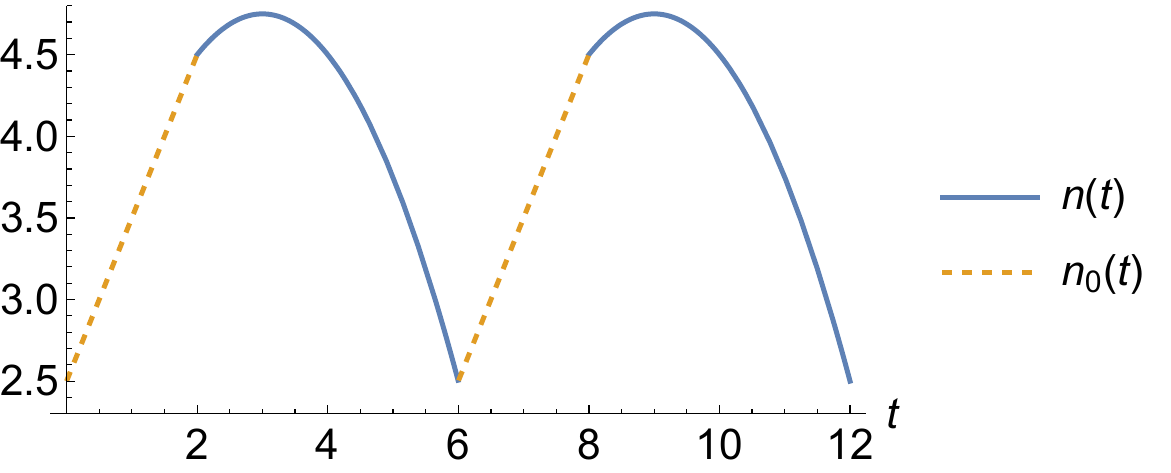}} \\
    \subcaptionbox{Price and utility versus time. \label{fig:stable_cycle_2_price_and_utility}}{\includegraphics[height = .33 \hsize]{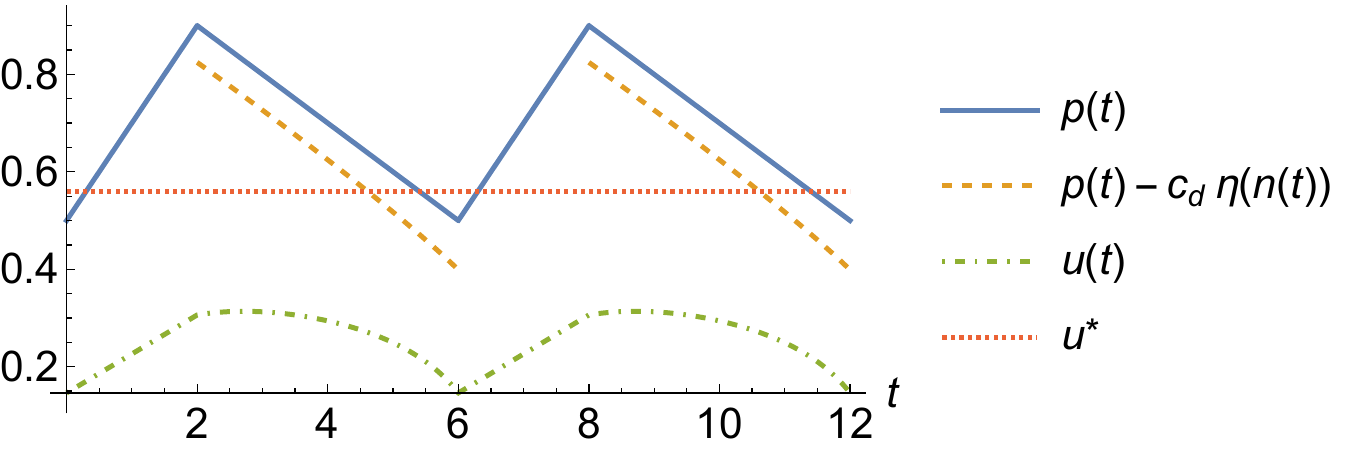}}
  \end{minipage}
  \caption{The ``asymmetric'' cycle in \Cref{exmp:cycles}, with $(\ell_+, \ell_-, \ubaralt{p}, \bar{p}, \hat{n}) = (0.2, 0.1, 0.5, 0.9, 4.75)$.  \label{fig:stable_cycle_2}}
\end{figure}

We formally establish the stability of these price cycles in the next section. Intuitively, drivers do not have an incentive to go online before others do as the price increases, since with a very small $n(t)$, drivers' en route time to pick up a rider will be very long. 
Further, if $c_d$ is appropriately high, going offline before other drivers do (as the price decreases) to wait for the next peak price is not better than remaining online and accepting dispatches from the platform.
\end{examplecont}

\section{Existence of Stable Price Cycles} \label{sec:cycle_existence}

The previous section illustrates %
how price cycles emerge when drivers collectively use the online/offline threshold strategies, and formally defines the stable price cycles induced by such strategies.
In this section, we study when they exist and characterize them. %
Specifically, we provide sufficient conditions
that can be easily satisfied and verified.
Further, we show that, while the cycles may form an equilibrium among drivers, they are not necessarily beneficial for drivers. In particular, we prove that stable cycles cannot generate a higher average utility for drivers than the socially optimal steady state, for markets that are sufficiently dense.

In order to prove the existence of stable price cycles, we return to \Cref{def:stable_cycles}, which has two conditions. Condition \ref{cond:repeat} on market clearance is easy to verify, while \ref{cond:driver} on drivers' best response is more difficult to approach. 
Thus, the first step in our analysis is to develop a sufficient condition for %
\ref{cond:driver} that is both easy to satisfy and easy to verify.

We start by introducing some %
notation. When supply and demand are uniformly distributed in space, dispatching the closest driver for each trip request corresponds to dispatching available drivers uniformly at random. %
The rate at which a particular driver is dispatched at time $t$ is therefore
\begin{equation} \label{eq:rho(t)}
    \rho(t) \triangleq \frac{\lambda_r \bar F(p(t) + c_r \eta(n(t)))}{n(t)}.
\end{equation}
Moreover, denote
\begin{equation} \label{eq:h(t)}
  h(t) \triangleq p(t) - c_d \eta(n(t)) - \frac{c_d}{\rho(t)}.
\end{equation}
$h(t)$ can be interpreted as the continuation payoff of a driver in a stationary environment where the price and number of drivers are fixed at $p(t)$ and $n(t)$, respectively.
Finally, let
\begin{equation}
    u(t) \triangleq \sup_\sigma u(t, \sigma, \sigma_{\bar p, \ubar p}) \label{eq:opt_cont_payoff}
\end{equation}
denote a driver's optimal continuation payoff at time $t$, when all other drivers use strategy $\sigma_{\bar p, \ubar p}$. The following lemma provides a characterization of $u(t)$.

\begin{lemma} \label{thm:opt_cont_payoff}
When other drivers adopt an online/offline strategy $\sigma_{\bar p, \ubar p}$, a driver's optimal continuation payoff $u(t)$ under the corresponding price cycle is continuous in $t$.
  For $t \in [t_0, t_1)$, $u(t) = u(t_1) - c_d (t_1 - t)$;
  for $t \in [t_1, t_2)$, %
  $u(t)$ is governed by the following differential equation:
  \begin{equation}
    \dot u(t) = c_d - \rho(t) \max\Set{0, p(t) - c_d \eta(N(t)) - u(t)}. \label{eq:-du}
  \end{equation}
\end{lemma}

The proof of this lemma is provided in \Cref{sec:pf_opt_cont_payoff}. 
Briefly, for $t \in [t_0, t_1)$, no driver %
has an incentive to sign online due to the long the pick-up time--- under $\sigma_{\bar p, \ubar p}$, the rest of the drivers all stay offline during $[t_0, t_1)$. %
$u(t)$ therefore increases at the rate of $c_d$, as the remaining time to wait until $t_1$ decreases. For $t \in [t_1, t_2)$,
the right hand side of $\eqref{eq:-du}$ can be interpreted as a driver optimizing between %
being online and offline %
at time $t$. Being offline results in a change rate of $c_d$ %
in the continuation payoff, while being online leads to an extra term, representing the possibility of being dispatched at the rate of $\rho(t)$, in which case the driver receives the net payoff from a trip $p(t) - c_d \eta(N(t))$ instead of the continuation payoff $u(t)$. %

Lemma~\ref{thm:opt_cont_payoff} allows us to establish the first main result of this section, which provides a sufficient condition for driver best-response under a price cycle.

\begin{theorem} \label{thm:suf_cond} 
Suppose a price cycle $(\ell_+, \ell_-, \ubar p, \bar p, \hat n)$ is market clearing, i.e., \Cref{def:stable_cycles} \ref{cond:repeat} holds. 
Then, \Cref{def:stable_cycles} \ref{cond:driver} holds, meaning that the online/offline strategy $\sigma_{\bar p, \ubar p}$ forms a Nash equilibrium among drivers and the cycle is stable, if 
the following two conditions are satisfied:
  \begin{align}
    \forall t \in (t_1, t_2),~ \frac{\dd}{\dd t} (p(t) - c_d \eta(n(t))) & < c_d, \label{eq:d(p_d)<c_d} \tag{C2.1} \\
    p(t_0) - c_d \eta(N(t_0)) + c_d (t_1 - t_0) & > \max_{t \in (t_1, t_2)} h(t). \tag{C2.2} \label{eq:strange_condition}
  \end{align}
\end{theorem}

When every other driver adopts strategy $\sigma_{\bar p, \ubar p}$, the strategy of an infinitesimal driver does not affect the number of online or offline drivers, %
or the trip prices %
determined by the platform.
For $t \in [t_0, t_1)$, we have argued in Lemma~\ref{thm:opt_cont_payoff} that signing online is not a useful deviation.
To prove that going offline at any point of time $t \in [t_1, t_2)$ is not a useful deviation from $\sigma_{\bar p, \ubar p}$ either,
we show that under \eqref{eq:d(p_d)<c_d} and \eqref{eq:strange_condition}, if a driver strictly prefers staying offline to accepting trips at some $t \in [t_1, t_2)$, then according to \Cref{thm:opt_cont_payoff}, $u(t)$ must strictly increase after one cycle period.
This contradicts the assumption that the price cycle is market clearing, under which the %
drivers' optimal continuation payoff should remain exactly the same after going through one period of the price cycle.
See \Cref{sec:pf_suf_cond} for the full proof of this theorem. 

This sufficient condition is easy to verify and holds broadly. Applying this result, we return to our running example and prove that the price cycles illustrated in Example~\ref{exmp:cycles} are stable.

\begin{examplecont}{exmp:cycles}
Both price cycles in \Cref{exmp:cycles} are stable. The full proof is provided in \Cref{appx:proof_prop_cycles_are_stable}, where we verify conditions \ref{cond:repeat}, \eqref{eq:d(p_d)<c_d}, and \eqref{eq:strange_condition}.
Note that the two examples are not the only stable price cycles in these markets. In fact, there can be infinitely many cycles that satisfy the conditions in \Cref{thm:suf_cond}, corresponding to different price thresholds, for example.
\end{examplecont}

Next, we turn to the question of whether drivers benefit from stable price cycles.  Perhaps surprisingly, the following result shows that the total payoff to drivers under any stable cycle is lower than that under the SOSS, under mild conditions.  In fact, drivers not only incur a higher total cost waiting for trip dispatches, but also receive a lower total net payoff from trips since many drivers are dispatched after the price had dropped below $p^\ast$. %

\begin{theorem} \label{thm:lower_payoff}
  Assume that the distribution of rider value satisfies
  \begin{equation} \label{eq:convex_reciprocal}
    \frac{\dd^2}{\dd v^2} \frac{1}{\bar F(v)} \ge 0,
  \end{equation}
  i.e., $\bar F(v)^{-1}$ is convex.
  The total payoff to drivers under any stable cycle is lower than that under the socially optimal steady state (SOSS), if
  \begin{equation} \label{eq:cond_SOSS_better}
    \frac{\ell_-}{\ell_+} + 1 > \frac{\bar F(p^* - c_d \eta(n^*))}{\bar F(p^* + c_r \eta(n^*))}.
  \end{equation}
\end{theorem}

Condition \eqref{eq:convex_reciprocal} holds for all the examples we present in this paper as well as for many common distributions, e.g., uniform, exponential, and Pareto with finite mean. 
Moreover, the difference between the numerator and denominator of the RHS of \eqref{eq:cond_SOSS_better} is driven by $c_d \eta(n^*)$ and $c_r \eta(n^*)$, the costs incurred by drivers and riders due to the en route time under the SOSS. 
Intuitively, for markets that are more dense, with higher $n^*$ and lower $\eta(n^*)$, the RHS of \eqref{eq:cond_SOSS_better} is closer to $1$ thus \eqref{eq:cond_SOSS_better} is more likely to hold.
For a typical market with $\ell_+ = \ell_-$, a violation of \eqref{eq:cond_SOSS_better} implies 
\[ 
    \frac{\ell_-}{\ell_+} + 1 = 2 < \frac{\bar F(p^* - c_d \eta(n^*))}{\bar F(p^* + c_r \eta(n^*))} \implies \bar F(p^* - c_d \eta(n^*)) > 2 \bar F(p^* + c_r \eta(n^*)).
\] 
What this means is that fixing the number of available drivers at $n^\ast$, rider demand will \emph{double} when trip price reduces from $p^\ast$ to $p^\ast - (c_r + c_d)\eta(n^*)$ (in which case riders effective cost of trip reduces from $p^* + c_r \eta(n^*)$ to $p^* - c_d \eta(n^*)$). This is not very likely for markets where the en route time is only a couple of minutes such that the total pick-up cost $(c_r + c_d)\eta(n^*)$ is relatively small. 

\begin{table}[h]
  \centering
  \caption{Comparison between the SOSS and the price cycles for the two running examples.}
  \label{tab:compare_utility}
  \begin{tabular}{l c c c c c}
    \toprule
    & \multicolumn{2}{c}{Symmetric example} & & \multicolumn{2}{c}{Asymmetric example} \\
    & SOSS & Cycle & & ~SOSS & Cycle \\
    \midrule
    Rider surplus rate & 0.1 & \textbf{0.2375} & & 0.1 & \textbf{0.1722} \\
    Driver surplus rate & \textbf{0.7325} & 0.1610 & & \textbf{0.56} & 0.2584 \\
    Social welfare rate & \textbf{0.8325} & 0.3985 & & \textbf{0.66} & 0.4306 \\
    \bottomrule
  \end{tabular}
\end{table}

\begin{examplecont}{exmp:cycles}
To see how much the drivers/riders/society as a whole is affected, we compute the average rider surplus and driver surplus per unit of time under the cycles and the socially optimal steady states for the two economies analyzed in \Cref{exmp:cycles}. 
As shown in \Cref{tab:compare_utility}, %
both driver surplus and social welfare are substantially lower. 
However, riders as a whole are better off. %
This is because %
the overall trip throughput
does not change, but trip prices drop substantially below the SOSS price during a cycle and the average %
trip price paid by riders who are picked up is lower. 
This is also due to the modeling assumption that riders are impatient. In practice, some riders will choose to wait if upon their arrival no available driver is online, and this will negatively affect the rider surplus.
\end{examplecont}

\begin{proof}[Proof of Theorem~\ref{thm:lower_payoff}]
It suffices to show that for any market that satisfies \eqref{eq:convex_reciprocal} and \eqref{eq:cond_SOSS_better}, drivers have lower utility in stable cycles than at the SOSS, i.e.,
  \begin{equation} \label{eq:cycle_low_utility}
  \begin{split}
    & \quad \int_{t_1}^{t_2} \lambda_r \bar F(p(t) + c_r \eta(n(t))) (p(t) - c_d \eta(n(t))) \dd t - \int_{t_0}^{t_2} c_d N(t) \dd t \\
    & < (\lambda_d (p^* - c_d \eta(n^*)) - c_d n^*) (t_2 - t_0),
  \end{split}
  \end{equation}
where the LHS is the total utility of drivers derive from one period of a cycle, and the RHS is the total utility generated at the SOSS during the same amount of time.
To this end, we show that the drivers have a higher total waiting cost \eqref{eq:cycle_high_cost}, as well as a lower total %
net payoff \eqref{eq:cycle_low_revenue} under any stable cycle:
\begin{align}
    (t_2 - t_0) c_d n^* & < \int_{t_0}^{t_2} c_d N(t) \dd t %
    \label{eq:cycle_high_cost} \\
    \int_{t_1}^{t_2} \lambda_r \bar F(p(t) + c_r \eta(n(t))) (p(t) - c_d \eta(n(t))) \dd t & < \lambda_d (p^* - c_d \eta(n^*)) (t_2 - t_0). \label{eq:cycle_low_revenue}
\end{align}

We first prove \eqref{eq:cycle_high_cost} by showing that $N(t) > n^*$ for any $t \in [t_0, t_2]$.

Drivers accumulate during $t \in [t_0, t_1]$, but there is zero net accumulation of drivers over the period of cycle. Thus we have $n(t_1) > N(t_0) = N(t_2)$. %
Recall that $N(t) = n(t)$ for all $t \in (t_1, t_2)$, we know $n(t)$ must be strictly decreasing at some time $t \in (t_1, t_2)$.

Between $t_1$ and $t_2$, the state evolution follows the intended trajectory under the platform's pricing rule. There are four quadrants in the state space, corresponding to $\dot p$ and $\dot n$ being positive/negative (as in \Cref{fig:pricing_policy}).
Once $\dot N(t) < 0$ (i.e., drivers are used up faster than replenished and the state is to the left of the bold red line in Figures~\ref{fig:soc_opt} and~\ref{fig:pricing_policy}), it will not be positive again until the price starts to increase.
Thus, $\dot N(t_2) < 0$ (left derivative), i.e., $(p(t_2), N(t_2))$ must be in the left region.

We also know that $(p(t_2), N(t_2))$ is in the upper region because the platform is still decreasing the price before drivers go offline.
Therefore, $(p(t_2), N(t_2))$ must be in the top left quadrant.
Moreover, the SOSS has the smallest $n$
among all states in the top left quadrant, so $(p^*, n^*)$ is the lowest point in the upper left region.
Thus, $N(t_2) > n^*$. By $n(t_1) > N(t_2)$ and the fact that $n(t)$ first increase and then decrease during $t \in [t_1, t_2]$, we know that for $t \in [t_1, t_2]$, $n(t) \ge N(t_2) > n^*$.

  \smallskip
  
 Below we prove \eqref{eq:cycle_low_revenue}.
 The SOSS satisfies that
  \begin{equation} \label{eq:SOSS}
    \lambda_d - \lambda_r \bar F(p^* + c_r \eta(n^*)) = 0.
  \end{equation}
  By \labelcref{eq:repeat,eq:SOSS},
  \begin{equation} \label{eq:connect_cycle&SOSS}
    \frac{\int_{t_1}^{t_2} \bar F(p(t) + c_r \eta(n(t))) \dd t}{t_2 - t_0}
    = \frac{\lambda_d}{\lambda_r}
    = \bar F(p^* + c_r \eta(n^*)).
  \end{equation}
  Define the effective cost of the trip for the riders as:
  $p_r(t) \triangleq p(t) + c_r \eta(n(t)).$ Then
  \begin{align*}
    \bar F(p^* - c_d \eta(n^*))
    < \frac{\ell_- + \ell_+}{\ell_+} \bar F(p^* + c_r \eta(n^*))
    = \frac{\ell_- + \ell_+}{\ell_+} \frac{\int_{t_1}^{t_2} \bar F(p_r(t)) \dd t}{t_2 - t_0}
    & = \frac{\int_{t_1}^{t_2} \bar F(p_r(t)) \dd t}{t_2 - t_1} \\
    & \le \bar F \left(\frac{\int_{t_1}^{t_2} \bar F(p_r(t)) p_r(t) \dd t}{\int_{t_1}^{t_2} \bar F(p_r(t)) \dd t}\right). 
  \end{align*}
  The first inequality follows from \eqref{eq:cond_SOSS_better}.
  The next equality follows from \eqref{eq:connect_cycle&SOSS}.
  The final step follows from \Cref{prop:FFFF}, which is proved in \Cref{appx:FFFF}.
  
  Then we apply $\bar F(\cdot)^{-1}$ on both sides.
  \[p^* - c_d \eta(n^*)
    > \frac{\int_{t_1}^{t_2} \bar F(p_r(t)) p_r(t) \dd t}{\int_{t_1}^{t_2} \bar F(p_r(t)) \dd t}
    \ge \frac{\int_{t_1}^{t_2} \bar F(p(t) + c_r \eta(n(t))) (p(t) - c_d \eta(n(t))) \dd t}{\int_{t_1}^{t_2} \bar F(p(t) + c_r \eta(n(t))) \dd t}.\]
  Moving the denominator in the right side to the left side and using \eqref{eq:repeat}, we have \eqref{eq:cycle_low_revenue}.
\end{proof}

To further illustrate this discussion, we give one more example with realistic parameters in \Cref{exmp:cycle_real} in the appendix.

\section{Avoiding Price Cycles} \label{sec:avoiding_cycles}

In the previous section, we showed that online/offline threshold strategies can induce stable price cycles, but may lead to a lower total utility for drivers in comparison to that under the SOSS.
Thus, an important question to address next is: what can a platform do to reduce the likelihood of price cycles emerging and/or minimize their impact?

There are three potentially natural approaches for reducing the impact of price cycles that have received attention. %
A first approach is to allow \emph{cherry-picking} based on trip details. %
One reason for a driver to not accept dispatches while the other drivers are offline is that, with few drivers online, the loss of density in space leads to a very long average pick-up time. Allowing drivers to freely decline trip dispatches (i.e., cherry-pick) partially mitigates this issue since an online driver may choose to only accept rider trips originating from close-by locations. However, incentivizing drivers to remain online in this way is not desirable in today's platforms since allowing frequent declines leads to excess strategizing on the part of drivers and a loss of reliability for riders.%
\footnote{In practice, many platforms discourages cherry-picking by pushing drivers offline for declining multiple trip dispatches~\citep{uberAutoOffline}.}

A second natural approach is to use prices that are more smooth in time, i.e. reducing $\ell_+$ and $\ell_-$. Intuitively, this smoothing reduces the incentive to go offline by increasing the time it takes for prices to rise, thereby the cost drivers incur when waiting for a better price. 
We show in \Cref{prop:max_price_range} that with smaller $\ell_+$ and $\ell_-$, the stable cycles are necessarily ``narrower'' with smaller $\bar p - \ubar p$. %
However, the degree of smoothing necessary to eliminate cycles with this approach is not practical. As long as $\ell_+ \geq c_d$, waiting offline can be profitable since the rate at which prices may increase is higher than drivers' opportunity cost for time. Reducing $\ell_+$ to below $c_d$ is not realistic, since when drivers make an average of \$20 per hour, $c_d \approx $ \$1/3 per minute. Restricting the price increase to less than \$1/3 per minute substantially limits the platform's ability to respond to changes in the market conditions in real-time, leading to significant market inefficiency.

Finally, the third %
approach to consider for reducing the likelihood and impact of price cycles is to introduce a \emph{price floor}.  
Intuitively, such a lower bound on trip prices prevents drivers' payoffs from remaining online from dropping too low, thereby reducing drivers' incentives to go offline and wait for a better price.
However, the fundamental question about the efficacy of this approach is: \emph{Can the price floor be set in a way the eliminates all stable price cycles? And, if so, where should the price floor be set?}

Our main result for this section is the following theorem, which characterizes the market conditions and the set of price floors that ensure no stable price cycles exist. 

\begin{theorem} \label{thm:platform_policy} Given any market condition that satisfies
  \begin{equation}
    \frac{\ell_-}{\ell_+} > \frac{\lambda_r}{\lambda_d} \bar F(p^*) - 1, \label{equ:floor_exists_2}
  \end{equation}
  the platform can ensure no stable price cycles exist%
\footnote{
Note that the theorem only %
rules out stable cycles we %
study in the paper, %
which correspond to online/offline strategies as in Definition~\ref{defn:online_offline_strategy}.
Technically, %
even with an appropriate price floor, it may form an equilibrium for drivers to coordinate on other more complicated collective strategies. However, %
the threshold strategies we study are aligned with news reports and driver forum discussions. Moreover, they do not require further communication among drivers once the thresholds are agreed upon (the only signal needed is the surge multiplier, which is provided by the app).
}
  and maintain feasibility of the %
  SOSS by setting a price floor $p_{\text{floor}}$ such that
  \begin{equation}
      \bar F^{-1} \left(\frac{(\ell_- + \ell_+) \lambda_d}{\ell_+ \lambda_r}\right) < p_{\text{floor}} < p^*. \label{equ:floor_exists}
  \end{equation}
\end{theorem}

This theorem provides a range of potential price floors that can eliminate the possibility of stable cycles for markets that satisfy \eqref{equ:floor_exists_2}. 
To interpret the market condition in \eqref{equ:floor_exists_2}, note from \eqref{eq:p^*} that at the SOSS $(p^\ast, n^\ast)$, we have $p^*  = \bar F^{-1}(\lambda_d / \lambda_r) - c_r \eta(n^*)$ . 
Therefore, for markets that are dense, i.e., markets with higher number of available drivers $n^\ast$ and lower pick-up time $\eta(n^*)$, $\frac{\lambda_r}{\lambda_d} \bar F(p^*) - 1$ is closer to 0, %
and \eqref{equ:floor_exists_2} is more likely to hold.

A violation of \eqref{equ:floor_exists_2} is not very likely in practice. Consider the typical platforms with $\ell_- = \ell_+$. A violation of \eqref{equ:floor_exists_2} implies that 
\begin{align}
    \frac{\ell_-}{\ell_+} = 1 < \frac{\lambda_r}{\lambda_d} \bar F(p^*) - 1 \implies \bar F(p^*) > 2 \frac{\lambda_d}{\lambda_r}.
\end{align}
Compared with $\bar F(p^* + c_r \eta(n^*)) = \lambda_d / \lambda_r$, we know that in a market that violates \eqref{equ:floor_exists_2}, eliminating the waiting time (i.e. reducing riders' effective cost for a trip from $p^* + c_r \eta(n^*)$ to $p^\ast$) will lead to a doubling of rider demand.
This is not likely for %
dense markets where riders wait for only a few minutes for the drivers to pick them up.

For a typical platform with $\ell_- = \ell_+$, observe that the smallest effective price floor $ \bar F^{-1} \left(\frac{(\ell_- + \ell_+) \lambda_d}{\ell_+ \lambda_r}\right)  =  \bar F^{-1} \left(\frac{2 \lambda_d}{\lambda_r}\right) $ is the effective cost of trips at which rider demand doubles in comparison to that under the SOSS. As a result, Theorem~\ref{thm:platform_policy} provides a simple rule of thumb, that a platform can eliminate stable cycles by imposing a lower bound on prices at the point where rider demand doubles. To demonstrate the use of price floors, we return to our running example.

\begin{examplecont}{exmp:cycles}
Consider the two economies studied in Example~\ref{exmp:cycles}.
Condition~\eqref{equ:floor_exists_2} holds because the RHS of \eqref{equ:floor_exists_2} is always zero due to the simpyfing assumption that $c_r = 0$.
For the first economy, illustrated in Figure~\ref{fig:stable_cycle}, any lower bound on prices $p_{\text{floor}}$ between $\bar F^{-1} \left(\frac{(\ell_- + \ell_+) \lambda_d}{\ell_+ \lambda_r}\right) = 0.6$ and $p^\ast = 0.8$ ensures that no stable price cycles may exist. 
For the other one in Figure~\ref{fig:stable_cycle_2}, with $\ell_+ = 0.2$ and $\ell_- = 0.1$, any price floor between $\bar F^{-1} \left(\frac{(\ell_- + \ell_+) \lambda_d}{\ell_+ \lambda_r}\right) = 0.7$ and $p^\ast = 0.8$ is effective.
\end{examplecont}

We now provide the proof of the main result of this section.

\begin{proof}[Proof of Theorem \ref{thm:platform_policy}]

To prove that no stable cycle exists when a price floor satisfying \eqref{equ:floor_exists} is imposed, we show that, given any stable cycle, the minimum price $\ubar{p}$ must be weakly below the left-hand side of \eqref{equ:floor_exists}.  This is true because otherwise the platform cannot dispatch all drivers that arrive in the region, violating the market clearing condition \Cref{cond:repeat}.

By condition \ref{cond:repeat} of \Cref{def:stable_cycles}, we know that the total driver arrival during one period
  \begin{align*}
    \lambda_d (t_2 - t_0)
    & = \int_{t_1}^{t_2} \lambda_r \bar F(p(t) + c_r \eta(n(t))) \dd t \le \lambda_r (t_2 - t_1) \bar F(p(\tilde t) + c_r \eta(n(\tilde t))),
  \end{align*}
  where $\tilde t = \arg \min_{t \in [t_1, t_2]} (p(t) + c_r \eta(n(t)))$. Further, since
  \[ (\bar{p}-\ubar{p}) = \ell_-(t_2 - t_1) = \ell_+(t_1 - t_0)
  \implies \frac{t_2 - t_0}{t_2 - t_1} = \frac{\ell_- + \ell_+}{\ell_+},\]
  we have
  \begin{alignat*}{2}
    & & \lambda_d (t_2 - t_0) & \le \lambda_r (t_2 - t_1) \bar F(p(\tilde t) + c_r \eta(n(\tilde t))) \\
    \implies {} & & \lambda_d (\ell_- + \ell_+) & \le \lambda_r \ell_+ \bar F(p(\tilde t) + c_r \eta(n(\tilde t))) \\
    \implies {} & & \bar F^{-1} \left(\frac{(\ell_- + \ell_+) \lambda_d}{\ell_+ \lambda_r} \right) & \ge p(\tilde t) + c_r \eta(n(\tilde t)) \ge p(\tilde t) \ge \ubar p.
  \end{alignat*}
As a result, no stable cycle exists when trip prices cannot drop below %
$p_{\text{floor}} > \bar F^{-1} \left(\frac{(\ell_- + \ell_+) \lambda_d}{\ell_+ \lambda_r}\right)$. %
Finally, market condition \eqref{equ:floor_exists_2} implies $\bar F^{-1} \left(\frac{(\ell_- + \ell_+) \lambda_d}{\ell_+ \lambda_r}\right) < p^\ast$, meaning that the platform is able to set a price floor that breaks all stable cycles while maintaining the feasibility of the SOSS price $p^\ast$.
\end{proof}

\section{Discussion \& Concluding Remarks} \label{sec:discussion}
\label{sec:conclusion}

Price cycles resulting from the collective online/offline strategies of drivers are a very visible source of inefficiency in today's ridesharing platforms, and undercut the reliability of service for riders. In this paper, we initiate an analytical study of such price cycles and discuss approaches for mitigating their impact. Our results show that the emergence of price cycles may form an equilibrium among drivers, but counter-intuitively leads to a lower total payoff to drivers in markets that are sufficiently dense. Further, we provide sufficient conditions to guide the design of price floors that effectively prevent the emergence of stable price cycles.

This paper is the first to provide a characterization of the existence and mitigation of price cycles and, as such, many important questions remain to be addressed. 
Our model focuses on a single region and assumes that drivers' arrival rate and opportunity cost are both exogenous. It would be interesting to understand the market dynamics in a metropolitan area with multiple regions supplied by a fixed pool of drivers, whose opportunity cost is driven by what can be earned on the platform.
Further, we have assumed that the arrival rates of drivers and riders are stationary and known to the platform. It is practically relevant to relax these assumptions and analyze the emergence of price cycles in a stochastic setting with incomplete information. %
Other interesting directions for future work include understanding the trade-offs surrounding other mitigation approaches such as ``price caps'', and incorporating multi-homing on the part of drivers and riders. 

More broadly, it is increasingly important to better understand control algorithms that make decisions based on the actions of the participants, when participants may coordinate on acting strategically. Consider, in particular, algorithms that \emph{learn} from the participants in the presence of incomplete information, e.g. DoorDash's raising driver pay for a job after each driver decline, assuming that drivers have better information on the desirability and/or cost of fulfilling the job. Such fine grained decisions based on the actions of fewer or even a single individual could potentially be more optimal but are necessarily more vulnerable, and it would be very interesting to quantify this trade-off and design mechanisms that are robust yet adaptive.

{\small
\bibliographystyle{plainnat}
\bibliography{price_cycles_refs.bib} 

\begin{thebibliography}{57}
\providecommand{\natexlab}[1]{#1}
\providecommand{\url}[1]{\texttt{#1}}
\expandafter\ifx\csname urlstyle\endcsname\relax
  \providecommand{\doi}[1]{doi: #1}\else
  \providecommand{\doi}{doi: \begingroup \urlstyle{rm}\Url}\fi

\bibitem[Afeche et~al.(2018)Afeche, Liu, and Maglaras]{afeche2018ride}
Philipp Afeche, Zhe Liu, and Costis Maglaras.
\newblock Ride-hailing networks with strategic drivers: The impact of platform
  control capabilities on performance.
\newblock \emph{Columbia Business School Research Paper}, \penalty0
  (18-19):\penalty0 18--19, 2018.

\bibitem[Ahmadinejad et~al.(2019)Ahmadinejad, Nazerzadeh, Saberi, Skochdopole,
  and Sweeney]{DBLP:conf/wine/AhmadinejadNSSS19}
AmirMahdi Ahmadinejad, Hamid Nazerzadeh, Amin Saberi, Nolan Skochdopole, and
  Kane Sweeney.
\newblock Competition in ride-hailing markets.
\newblock In Ioannis Caragiannis, Vahab~S. Mirrokni, and Evdokia Nikolova,
  editors, \emph{Web and Internet Economics - 15th International Conference,
  {WINE} 2019, New York, NY, USA, December 10-12, 2019, Proceedings}, volume
  11920 of \emph{Lecture Notes in Computer Science}, page 333. Springer, 2019.
\newblock URL
  \url{https://link.springer.com/content/pdf/bbm\%3A978-3-030-35389-6\%2F1.pdf}.

\bibitem[Aouad and Sarita{\c{c}}(2020)]{aouad2020dynamic}
Ali Aouad and {\"O}mer Sarita{\c{c}}.
\newblock Dynamic stochastic matching under limited time.
\newblock In \emph{Proceedings of the 21st ACM Conference on Economics and
  Computation}, pages 789--790, 2020.

\bibitem[Asadpour et~al.(2019)Asadpour, Lobel, and van
  Ryzin]{asadpour2019minimum}
Arash Asadpour, Ilan Lobel, and Garrett van Ryzin.
\newblock Minimum earnings regulation and the stability of marketplaces.
\newblock \emph{Available at SSRN 3502607}, 2019.

\bibitem[Ashlagi et~al.(2018)Ashlagi, Burq, Dutta, Jaillet, Saberi, and
  Sholley]{ashlagi2018maximum}
Itai Ashlagi, Maximilien Burq, Chinmoy Dutta, Patrick Jaillet, Amin Saberi, and
  Chris Sholley.
\newblock Maximum weight online matching with deadlines.
\newblock \emph{arXiv preprint arXiv:1808.03526}, 2018.

\bibitem[Banerjee et~al.(2015)Banerjee, Johari, and
  Riquelme]{BanJohRiq2015Pricing}
Siddhartha Banerjee, Ramesh Johari, and Carlos Riquelme.
\newblock Pricing in ride-sharing platforms: {A} queueing-theoretic approach.
\newblock In \emph{{ACM} Conference on Economics and Computation {(EC)}}, 2015.

\bibitem[Banerjee et~al.(2018)Banerjee, Kanoria, and Qian]{banerjee2018state}
Siddhartha Banerjee, Yash Kanoria, and Pengyu Qian.
\newblock State dependent control of closed queueing networks.
\newblock In \emph{Abstracts of the 2018 ACM International Conference on
  Measurement and Modeling of Computer Systems}, 2018.

\bibitem[Besbes et~al.(2018)Besbes, Castro, and Lobel]{besbes2018spatial}
Omar Besbes, Francisco Castro, and Ilan Lobel.
\newblock Spatial capacity planning.
\newblock \emph{Available at SSRN 3292651}, 2018.

\bibitem[Besbes et~al.(2020)Besbes, Castro, and Lobel]{besbes2020surge}
Omar Besbes, Francisco Castro, and Ilan Lobel.
\newblock Surge pricing and its spatial supply response.
\newblock \emph{Management Science}, 2020.

\bibitem[Bimpikis et~al.(2019)Bimpikis, Candogan, and
  Saban]{bimpikis2019spatial}
Kostas Bimpikis, Ozan Candogan, and Daniela Saban.
\newblock Spatial pricing in ride-sharing networks.
\newblock \emph{Operations Research}, 67\penalty0 (3):\penalty0 744--769, 2019.

\bibitem[Cai et~al.(2019)Cai, Bose, and Wierman]{cai2019role}
Desmond Cai, Subhonmesh Bose, and Adam Wierman.
\newblock On the role of a market maker in networked cournot competition.
\newblock \emph{Mathematics of Operations Research}, 44\penalty0 (3):\penalty0
  1122--1144, 2019.

\bibitem[Castillo(2020)]{castillo2020benefits}
Juan~Camilo Castillo.
\newblock Who benefits from surge pricing?
\newblock \emph{Available at SSRN 3245533}, 2020.

\bibitem[Castillo et~al.(2017)Castillo, Knoepfle, and Weyl]{castillo2017surge}
Juan~Camilo Castillo, Dan Knoepfle, and Glen Weyl.
\newblock Surge pricing solves the wild goose chase.
\newblock In \emph{Proceedings of the 2017 ACM Conference on Economics and
  Computation}, pages 241--242, 2017.

\bibitem[Castro et~al.(2020)Castro, Frazier, Ma, Nazerzadeh, and
  Yan]{castro2020matching}
Francisco Castro, Peter Frazier, Hongyao Ma, Hamid Nazerzadeh, and Chiwei Yan.
\newblock Matching queues, flexibility and incentives.
\newblock \emph{arXiv preprint arXiv:2006.08863}, 2020.

\bibitem[Castro et~al.(2021)Castro, Ma, Nazerzadeh, and
  Yan]{castro2021randomized}
Francisco Castro, Hongyao Ma, Hamid Nazerzadeh, and Chiwei Yan.
\newblock Randomized fifo mechanisms.
\newblock \emph{Technical report, Columbia University}, 2021.

\bibitem[Chen et~al.(2020)Chen, Ding, List, and Mogstad]{chen2020reservation}
Kuan-Ming Chen, Claire Ding, John~A List, and Magne Mogstad.
\newblock Reservation wages and workers’ valuation of job flexibility:
  Evidence from a natural field experiment.
\newblock Technical report, National Bureau of Economic Research, 2020.

\bibitem[Chen and Sheldon(2015)]{chen2015dynamic}
M~Keith Chen and Michael Sheldon.
\newblock Dynamic pricing in a labor market: Surge pricing and flexible work on
  the uber platform.
\newblock 2015.
\newblock URL
  \url{https://www.anderson.ucla.edu/faculty_pages/keith.chen/papers/SurgeAndFlexibleWork_WorkingPaper.pdf}.

\bibitem[Chen et~al.(2019)Chen, Rossi, Chevalier, and Oehlsen]{chen2019value}
M~Keith Chen, Peter~E Rossi, Judith~A Chevalier, and Emily Oehlsen.
\newblock The value of flexible work: Evidence from uber drivers.
\newblock \emph{Journal of Political Economy}, 127\penalty0 (6):\penalty0
  2735--2794, 2019.

\bibitem[Cohen et~al.(2016)Cohen, Hahn, Hall, Levitt, and
  Metcalfe]{cohen2016using}
Peter Cohen, Robert Hahn, Jonathan Hall, Steven Levitt, and Robert Metcalfe.
\newblock Using big data to estimate consumer surplus: The case of uber.
\newblock \emph{NBER Working Paper No. 22627}, 2016.

\bibitem[Cook et~al.(2018)Cook, Diamond, Hall, List, and Oyer]{cook2018gender}
Cody Cook, Rebecca Diamond, Jonathan Hall, John~A List, and Paul Oyer.
\newblock The gender earnings gap in the gig economy: Evidence from over a
  million rideshare drivers.
\newblock Technical report, National Bureau of Economic Research, 2018.

\bibitem[Cramer and Krueger(2016)]{cramer2016disruptive}
Judd Cramer and Alan~B Krueger.
\newblock Disruptive change in the taxi business: The case of uber.
\newblock \emph{American Economic Review}, 106\penalty0 (5):\penalty0 177--82,
  2016.

\bibitem[Dholakia(2015)]{dholakia2015everyone}
Utpal~M Dholakia.
\newblock Everyone hates uber’s surge pricing--here’s how to fix it.
\newblock \emph{Harvard Business Review}, 21\penalty0 (December), 2015.

\bibitem[Dickerson et~al.(2018)Dickerson, Sankararaman, Srinivasan, and
  Xu]{dickerson2018allocation}
John Dickerson, Karthik Sankararaman, Aravind Srinivasan, and Pan Xu.
\newblock Allocation problems in ride-sharing platforms: Online matching with
  offline reusable resources.
\newblock In \emph{Proceedings of the AAAI Conference on Artificial
  Intelligence}, volume~32, 2018.

\bibitem[{Dustin is Driving}(2019)]{dustinTalk}
{Dustin is Driving}.
\newblock Uber- i can't beleive drivers admitted to doing ``surge club'', 2019.
\newblock URL \url{https://youtu.be/SYkLhXMsZ8I}.
\newblock [Online; accessed 31-January-2021].

\bibitem[Fang et~al.(2019)Fang, Huang, and Wierman]{fang2019prices}
Zhixuan Fang, Longbo Huang, and Adam Wierman.
\newblock Prices and subsidies in the sharing economy.
\newblock \emph{Performance Evaluation}, 2019.

\bibitem[Fang et~al.(2020)Fang, Huang, and Wierman]{fang2020loyalty}
Zhixuan Fang, Longbo Huang, and Adam Wierman.
\newblock Loyalty programs in the sharing economy: Optimality and competition.
\newblock \emph{Performance Evaluation}, 2020.

\bibitem[Ford(2021)]{doordashGaming}
Broody Ford.
\newblock Doordash drivers game algorithm to increase pay, 2021.
\newblock URL
  \url{https://www.bloomberg.com/news/articles/2021-04-06/doordash-workers-are-trying-to-game-the-algorithm-to-increase-pay}.
\newblock [Online; accessed 6-August-2021].

\bibitem[Freund and van Ryzin(2021)]{freund2021pricing}
Daniel Freund and Garrett van Ryzin.
\newblock Pricing fast and slow: Limitations of dynamic pricing mechanisms in
  ride-hailing.
\newblock \emph{Available at SSRN 3931844}, 2021.

\bibitem[Garg and Nazerzadeh(2020{\natexlab{a}})]{garg2019driver}
Nikhil Garg and Hamid Nazerzadeh.
\newblock Driver surge pricing.
\newblock \emph{Proceedings of the 21st ACM Conference on Economics and
  Computation}, 2020{\natexlab{a}}.

\bibitem[Garg and Nazerzadeh(2020{\natexlab{b}})]{garg2020driver}
Nikhil Garg and Hamid Nazerzadeh.
\newblock Driver surge pricing.
\newblock In \emph{Proceedings of the 21st ACM Conference on Economics and
  Computation}, pages 501--501, 2020{\natexlab{b}}.

\bibitem[Hall et~al.(2015)Hall, Kendrick, and Nosko]{hall2015effects}
Jonathan Hall, Cory Kendrick, and Chris Nosko.
\newblock The effects of uber's surge pricing: A case study.
\newblock Technical report, The University of Chicago Booth School of Business,
  2015.

\bibitem[Hall and Krueger(2016)]{hall2016analysis}
Jonathan~V Hall and Alan~B Krueger.
\newblock An analysis of the labor market for uber’s driver-partners in the
  united states.
\newblock \emph{NBER Working Paper No. 22843}, 2016.

\bibitem[Hall et~al.(2017)Hall, Horton, and Knoepfle]{hall2017labor}
Jonathan~V Hall, John~J Horton, and Daniel~T Knoepfle.
\newblock Labor market equilibration: Evidence from uber.
\newblock Technical report, New York University Stern School of Business, 2017.

\bibitem[Kanoria and Qian(2020)]{kanoria2020blind}
Yash Kanoria and Pengyu Qian.
\newblock Blind dynamic resource allocation in closed networks via mirror
  backpressure.
\newblock In \emph{EC'20: Proceedings of the 21st ACM Conference on Economics
  and Computation}, 2020.

\bibitem[Lian and van Ryzin(2020)]{lian2020autonomous}
Zhen Lian and Garrett van Ryzin.
\newblock Autonomous vehicle market design.
\newblock \emph{Available at SSRN}, 2020.

\bibitem[Lian and Van~Ryzin(2021)]{lian2021optimal}
Zhen Lian and Garrett Van~Ryzin.
\newblock Optimal growth in two-sided markets.
\newblock \emph{Management Science}, 2021.

\bibitem[Lian et~al.(2021)Lian, Martin, and van Ryzin]{lian2021larger}
Zhen Lian, Sebastien Martin, and Garrett van Ryzin.
\newblock Larger firms pay more in the gig economy.
\newblock \emph{Available at SSRN}, 2021.

\bibitem[Lu et~al.(2018)Lu, Frazier, and Kislev]{lu2018surge}
Alice Lu, Peter~I Frazier, and Oren Kislev.
\newblock Surge pricing moves uber's driver-partners.
\newblock In \emph{Proceedings of the 2018 ACM Conference on Economics and
  Computation}, pages 3--3, 2018.

\bibitem[Lyft(2017)]{lyftmission}
Lyft.
\newblock Pricing when it's busy, 2017.
\newblock URL
  \url{https://web.archive.org/web/20170216214743/https://help.lyft.com/hc/en-us/articles/213818898-Prime-Time-for-Passengers}.
\newblock [Online; accessed 16-February-2017].

\bibitem[Lyft(2021)]{lyft2021acceptance}
Lyft.
\newblock Acceptance rate - lyft help, 2021.
\newblock URL
  \url{https://web.archive.org/web/20210131022046/https://help.lyft.com/hc/en-us/articles/115013077708-Acceptance-rate}.
\newblock [Online; accessed 30-January-2021].

\bibitem[Ma et~al.(2019)Ma, Fang, and Parkes]{ma2019spatio}
Hongyao Ma, Fei Fang, and David~C Parkes.
\newblock Spatio-temporal pricing for ridesharing platforms.
\newblock In \emph{Proceedings of the 2019 ACM Conference on Economics and
  Computation}, pages 583--583, 2019.

\bibitem[Marshall(2020)]{changeRuleAdjustStrategy}
Aarian Marshall.
\newblock Uber changes its rules, and drivers adjust their strategies, 2020.
\newblock URL
  \url{https://web.archive.org/web/20200219030836/https://www.wired.com/story/uber-changes-rules-drivers-adjust-strategies/amp}.
\newblock [Online; accessed 19-February-2020].

\bibitem[M{\"o}hlmann and Zalmanson(2017)]{mohlmann2017hands}
Marieke M{\"o}hlmann and Lior Zalmanson.
\newblock Hands on the wheel: Navigating algorithmic management and uber
  drivers’.
\newblock \emph{38th ICIS Proceedings}, 2017.

\bibitem[Ostrovsky and Schwarz(2019)]{ostrovsky2019carpooling}
Michael Ostrovsky and Michael Schwarz.
\newblock Carpooling and the economics of self-driving cars.
\newblock In \emph{Proceedings of the 2019 ACM Conference on Economics and
  Computation}, pages 581--582, 2019.

\bibitem[{\"O}zkan and Ward(2020)]{ozkan2020dynamic}
Erhun {\"O}zkan and Amy~R Ward.
\newblock Dynamic matching for real-time ride sharing.
\newblock \emph{Stochastic Systems}, 10\penalty0 (1):\penalty0 29--70, 2020.

\bibitem[Pang et~al.(2017)Pang, Fu, Lee, and Wierman]{pang2017efficiency}
John~ZF Pang, Hu~Fu, Won~I Lee, and Adam Wierman.
\newblock The efficiency of open access in platforms for networked cournot
  markets.
\newblock In \emph{IEEE INFOCOM 2017}, pages 1--9, 2017.

\bibitem[Qin et~al.(2020)Qin, Tang, Jiao, Zhang, Xu, Zhu, and Ye]{qin2020ride}
Zhiwei Qin, Xiaocheng Tang, Yan Jiao, Fan Zhang, Zhe Xu, Hongtu Zhu, and
  Jieping Ye.
\newblock Ride-hailing order dispatching at {DiDi} via reinforcement learning.
\newblock \emph{INFORMS Journal on Applied Analytics}, 50\penalty0
  (5):\penalty0 272--286, 2020.

\bibitem[Rayle et~al.(2014)Rayle, Shaheen, Chan, Dai, and
  Cervero]{rayle2014app}
Lisa Rayle, Susan Shaheen, Nelson Chan, Danielle Dai, and Robert Cervero.
\newblock App-based, on-demand ride services: Comparing taxi and ridesourcing
  trips and user characteristics in san francisco.
\newblock Technical report, University of California Transportation Center
  (UCTC), 2014.

\bibitem[Rheingans-Yoo et~al.(2019)Rheingans-Yoo, Kominers, Ma, and
  Parkes]{rheingans2019ridesharing}
Duncan Rheingans-Yoo, Scott~Duke Kominers, Hongyao Ma, and David~C Parkes.
\newblock Ridesharing with driver location preferences.
\newblock In \emph{Proceedings of the 28th International Joint Conference on
  Artificial Intelligence}, 2019.

\bibitem[Sweeney(2019)]{WJLA2019ReganNationalCollusion}
Sam Sweeney.
\newblock Uber, lyft drivers manipulate fares at reagan national causing
  artificial price surges, 2019.
\newblock URL
  \url{https://web.archive.org/web/20210131171241/https://wjla.com/news/local/uber-and-lyft-drivers-fares-at-reagan-national}.
\newblock [Online; accessed 31-January-2021].

\bibitem[the Hood(2019)]{uberDriverSurge}
Uber~Under the Hood.
\newblock What moves us: Shedding more light on new driver surge, 2019.
\newblock URL
  \url{https://medium.com/uber-under-the-hood/what-moves-us-shedding-more-light-on-new-driver-surge-d580488a6761}.
\newblock [Online; accessed 6-August-2021].

\bibitem[Uber(2016)]{ubermissionChicago}
Uber.
\newblock Uber community guidelines, 2016.
\newblock URL
  \url{https://web.archive.org/web/20201025200022/https://www.uber.com/blog/chicago/uberaccess/}.
\newblock [Online; accessed 25-October-2020].

\bibitem[Uber(2017)]{uberAutoOffline}
Uber.
\newblock Auto-offline has changed for the better, 2017.
\newblock URL
  \url{https://www.uber.com/en-GB/blog/auto-offline-has-changed-for-the-better/}.
\newblock [Online; accessed 31-January-2021].

\bibitem[{UberPeople.net users}(2015)]{uberPeopleLondon}
{UberPeople.net users}.
\newblock Log off until surge!!!, 2015.
\newblock URL
  \url{https://web.archive.org/web/20210810221500/https://www.uberpeople.net/threads/log-off-until-surge.44164/}.
\newblock [Online; accessed 10-August-2021].

\bibitem[{UberPeople.net users}(2016)]{uberPeopleLA}
{UberPeople.net users}.
\newblock Stay logged off at lax, until it surges, 2016.
\newblock URL
  \url{https://web.archive.org/web/20210810222756/https://www.uberpeople.net/threads/stay-logged-off-at-lax-until-it-surges.102047/}.
\newblock [Online; accessed 10-August-2021].

\bibitem[Xu et~al.(2020)Xu, AMC~Vignon, Yin, and Ye]{xu2020empirical}
Zhengtian Xu, Daniel AMC~Vignon, Yafeng Yin, and Jieping Ye.
\newblock An empirical study of the labor supply of ride-sourcing drivers.
\newblock \emph{Transportation Letters}, pages 1--4, 2020.

\bibitem[Yan et~al.(2020)Yan, Zhu, Korolko, and Woodard]{yan2020dynamic}
Chiwei Yan, Helin Zhu, Nikita Korolko, and Dawn Woodard.
\newblock Dynamic pricing and matching in ride-hailing platforms.
\newblock \emph{Naval Research Logistics (NRL)}, 67\penalty0 (8):\penalty0
  705--724, 2020.

\end{thebibliography}
}

\newpage

\appendix

\section{Notation Table} \label{sec:notations}

\begin{longtable}{ll}
  \caption{List of notations.}
  \label{tab:notations} \\
    \toprule
    $\lambda_d$ & driver arrival rate \\
    $\lambda_r$ & rider arrival rate \\
    $c_d$ & driver's cost of time \\
    $c_r$ & rider's cost of time \\
    $p$ & price of a trip \\
    $n$ & number of online drivers \\
    $n_0$ & number of offline drivers \\
    $N$ & total number of drivers \\
    $\eta(n) = \tau n^{-\alpha}$ & en route time (pick-up time) when there are $n$ online drivers \\
    $V$ & random variable: a rider's value for a trip \\
    $\bar F(\cdot)$ & complementary c.d.f.\ of rider's value for a trip \\
    $w$ & social welfare per unit of time \\
    $p^*$ & price at the socially optimal steady state \\
    $n^*$ & number of drivers at the socially optimal steady state \\
    $\ell_+$ & maximum increase rate of price \\
    $\ell_-$ & maximum decrease rate of price \\
    $\sigma$ & driver's strategy \\
    $\ubar p$ & price at which drivers drop offline; the minimum price in a cycle \\
    $\bar p$ & price at which drivers join online; the maximum price in a cycle \\
    $\hat n$ & peak number of drivers in the cycle \\
    $u^*$ & driver's continuation payoff (utility) at the socially optimal steady state \\
    $u$ & driver's continuation payoff (utility) \\
    $t_0$ & the moment when drivers drop offline \\
    $t_1$ & the moment when drivers join online after $t_0$ \\
    $t_2$ & the moment when drivers drop offline again after $t_1$ \\
    $T$ & period of the cycle \\
    \bottomrule
\end{longtable}

\section{Proof of Lemma~\ref{thm:opt_cont_payoff}} \label{sec:pf_opt_cont_payoff}

With the strategy of other drivers fixed, the strategy of an infinitesimal driver do not affect the number of online or offline drivers, $n(t), n_0(t)$, or the trip prices $p(t)$ determined by the platform.
If all drivers use the online/offline strategy $\sigma_{\bar p, \ubar p}$, all drivers stay offline when the price rises (i.e., $n(t) = 0$ for all $t \in [t_0, t_1)$).
Due to the infinitely long en route time, a driver has no incentive to deviate from $\sigma_{\bar p, \ubar p}$, sign online, and accept dispatches.
Thus, for $t \in [t_0, t_1)$, the continuation payoff at time $t$ is the one at time $t_1$ minus the waiting cost from $t$ to $t_1$.

For each time interval $[s, t) \subseteq [t_1, t_2)$, a driver can stay offline, giving
\begin{equation}
  u(s) \ge u(t) - c_d (t - s),
\end{equation}
or stay online, giving
\begin{equation}
  \begin{aligned}
    u(s) & \ge q_{s,t} \E[p(\xi_{s,t}) - c_d \eta(n(\xi_{s,t})) - c_d (\xi_{s,t} - s)] + (1 - q_{s,t}) (u(t) - c_d (t - s)) \\
    & = u(t) - c_d (t - s) + q_{s,t} \E[p(\xi_{s,t}) - c_d \eta(n(\xi_{s,t})) - u(t) + c_d (t - \xi_{s,t})],
  \end{aligned}
\end{equation}
where $q_{s,t}$ is the probability of being dispatched during $[s, t)$, and $\xi_{s,t}$ is the dispatch time --- a random variable supported on $[s, t)$.
Combine the two inequalities to get
\begin{equation} \label{eq:s-t-ineq}
  u(s) \ge u(t) - c_d (t - s) + q_{s,t} \max\Set{0, \E[p(\xi_{s,t}) - c_d \eta(n(\xi_{s,t})) - u(t) + c_d (t - \xi_{s,t})]}.
\end{equation}
The equality holds if the driver keeps online or offline during $[s, t)$ without switching.

We assume that during any time interval $[s, t)$, a driver can only switch the online/offline status finitely many times.

Fix $t$ and let $s \nearrow t$. (Or fix $s$ and let $t \searrow s$. The former one gives the left limit, while the latter one gives the right limit. They are the same if $s, t \in [t_1, t_2]$.) Due to the above assumption, there is no status change during $[s, t)$ if $s$ is sufficiently close to $t$, and thus \eqref{eq:s-t-ineq} becomes an equality:
\begin{equation} \label{eq:s-t-eq}
  u(s) = u(t) - c_d (t - s) + q_{s,t} \max\Set{0, \E[p(\xi_{s,t}) - c_d \eta(n(\xi_{s,t})) - u(t) + c_d (t - \xi_{s,t})]}.
\end{equation}

The dispatch rate at time $r$ can be written as:
\[\rho(r) = \frac{-\frac{\dd}{\dd r} \Pr[t_\text{dispatch} \ge r]}{\Pr[t_\text{dispatch} \ge r]}.\]
Solve this ODE with boundary condition $\Pr[t_\text{dispatch} \ge s] = 1$ to get
\[\Pr[t_\text{dispatch} \ge r] = \exp(- \int_s^r \rho(x) \dd x).\]
Thus,
\[q_{s,t} = 1 - \Pr[t_\text{dispatch} \ge t] = 1 - \exp(- \int_s^t \rho(x) \dd x) \le \int_s^t \rho(x) \dd x \le \frac{\lambda_r}{n^*} (t - s),\]
where the last inequality follows from \eqref{eq:rho(t)}, the definition of $\rho$, with $n(t) \ge n^*$ for $t \in [t_1, t_2)$.
This implies that as $t - s \to 0$, we have $q_{s,t} \to 0$, and further by \eqref{eq:s-t-eq}, $u(s) - u(t) \to 0$. In other words, $u(t)$ is continuous in $[t_1, t_2]$.

\eqref{eq:s-t-eq} also implies that
\[\frac{u(s) - u(t)}{t - s} = -c_d + \frac{q_{s,t}}{t - s} \max\Set{0, \E[p(\xi_{s,t}) - c_d \eta(n(\xi_{s,t})) - u(t) + c_d (t - \xi_{s,t})]}.\]
Taking the left and right limits, we have
\begin{equation*}
  {- \dot u(t)} = - c_d + \rho(t) \max\Set{0, p(t) - c_d \eta(N(t)) - u(t)}.
\end{equation*}
This derivative represents the right derivative at $t_1$, the left derivative at $t_2$, and both left and right derivatives at $t \in (t_1, t_2)$. We use $N(t)$ instead of $n(t)$ to accommodate the case of $t = t_2$, because $\lim_{\xi \nearrow t_2} n(\xi) = N(t_2) \ne n(t_2) = 0$.

\section{Proof of Theorem \ref{thm:suf_cond}} \label{sec:pf_suf_cond}

For $t \in [t_1, t_2]$, let
\[g(t) \triangleq p(t) - c_d \eta(N(t)) - u(t)\]
be the difference between a driver's net earning $p(t) - c_d \eta(N(t))$ from an immediate trip (i.e. if she accepts a trip dispatch that she has just received), and the driver's optimal continuation payoff $u(t)$ (if she does not have a trip dispatch in hand).
We know $g(t)$ is continuous on $[t_1, t_2)$. When $g(t) \ge 0$, the driver's best-response is to remain online %
and accept dispatches from the platform if she receives one.
From \labelcref{eq:d(p_d)<c_d,eq:-du}, we have
\begin{equation}
  \dot g(t) < \max\{0, \rho(t) g(t)\},~\forall t \in (t_1, t_2). \label{eq:dg}
\end{equation}

We claim that $g(t_2) > 0$ implies $g(t) > 0$ for all $t\in [t_1, t_2)$, meaning that deviating from $\sigma_{\bar p, \ubar p}$ and going offline at times $t \in [t_1, t_2]$ is not useful. This completes the proof that $\sigma_{\bar p, \ubar p}$ forms a Nash equilibrium among drivers.
To see this, assume towards a contradiction that there exists $t' \in [t_1, t_2)$ such that $g(t') \leq 0$. With \eqref{eq:dg}, we know $\dot g(t') < 0$. As a result, $g(t) < 0$ and $\dot g(t) < 0$ hold for all $t \in [t', t_2)$. Further, %
the continuity of $g$ implies $g(t_2) < 0$, which contradicts %
$g(t_2) > 0$.

\smallskip

What is left to show is $g(t_2) > 0$. We prove this by contradiction, showing that if $g(t_2) \leq 0$, then, given the two conditions in the theorem, a driver's continuation payoff will be strictly higher after one cycle period, i.e., $u(t_0) < u(t_2)$. This cannot happen since the cycle repeats itself, and thus $u(t_2) = u(t_0)$ must hold. We discuss the cases of $g(t_2) = 0$ and $g(t_2) < 0$ separately.
\newcommand{\uhat}{\hat{u}}
\newcommand{\utilde}{\tilde{u}}
\newcommand{\gtilde}{\tilde{g}}
\newcommand{\ghat}{\hat{g}}

\paragraph{Case 1: $g(t_2) = 0$.}
Consider any function $\utilde(\cdot)$ such that $\utilde(t_2) = p(t_2) - c_d \eta(N(t_2))$. 
If $\utilde(t)$ represents the optimal continuation payoff of a driver at time $t$ when $\sigma_{\bar p, \ubar p}$ is adopted by every other driver, then (i) Equation \eqref{eq:-du} is satisfied for all $t \in (t_1, t_2)$, (ii) $\utilde(t+T) = \utilde(t)$ for all $t$, and (iii) $\utilde(t)$ is continuous.
Define
\[\gtilde(t) \triangleq p(t) - c_d \eta(N(t)) -\utilde(t).\]
We know that $\gtilde(t)$ is continuous on $[t_1, t_2]$, $\gtilde(t_2) = 0$, and that \eqref{eq:dg} holds for $\gtilde$ for all $t \in (t_1, t_2)$.
We claim that for all $t \in [t_1, t_2)$, $\gtilde(t) > 0$.
Otherwise, if $\gtilde(t') \le 0$ for some time $t' \in [t_1, t_2)$, then by \eqref{eq:dg}, $\frac{\dd}{\dd t} \gtilde(t) \big|_{t=t'} < 0$, and thus for all $t \in (t', t_2)$, we have $\gtilde(t), \frac{\dd}{\dd t} \gtilde(t) < 0$. The continuity of $\gtilde(t)$ then implies $\gtilde(t_2) < 0$, contradicting the assumption that $\gtilde(t_2) = 0$.

Given $\gtilde(t) > 0$ for all $t \in [t_1, t_2)$, \eqref{eq:-du} can be simplified as
\begin{equation} \label{eq:-du_sim}
  {-\frac{\dd}{\dd t} \utilde(t)}
  = \rho(t) (p(t) - c_d \eta(n(t)) - \utilde(t)) - c_d.
\end{equation}
When $\utilde(t) = h(t)$ (as defined in \eqref{eq:h(t)}), the RHS of \eqref{eq:-du_sim} is zero. As a result, for any $t \in (t_1, t_2)$,
\[\utilde(t) > h(t) \iff - \frac{\dd}{\dd t} \utilde(t) < 0.\]

If $\utilde(t_1) \geq \utilde(t_2)$, then $- \frac{\dd}{\dd t} \utilde(t) \geq 0$ must hold for some $t \in [t_1, t_2)$. This means that we must have $\utilde(t) \leq h(t)$ at some $t \in [t_1, t_2)$.
Now consider $t' \triangleq \min \{t \in [t_1, t_2) \given \utilde(t) \leq h(t) \}$.
The continuity of $\utilde(t)$ and $h(t)$ implies that $t'$ exists and that $\utilde(t') = h(t')$.
For all $\tau \in [t_1, t')$, we have $\utilde(t) > h(t)$ and $- \frac{\dd}{\dd t} \utilde(t) < 0$. 
This implies that
\[\utilde(t_1) \leq \utilde(t') = h(t') \leq \max_{t \in [t_1, t_2]} h(t).\]
Therefore, either $\utilde(t_1) \le \max_{t \in [t_1, t_2]} h(t)$ or $\utilde(t_1) < \utilde(t_2)$.
With \eqref{eq:strange_condition}, i.e. $\max_{t \in [t_1, t_2]} h(t) < p(t_0) - c_d \eta(N(t_0)) + c_d (t_1 - t_0)$, we know either
\[\utilde(t_1) < p(t_0) - c_d \eta(N(t_0)) + c_d (t_1 - t_0)\]
or
\[\utilde(t_1) < \utilde(t_2) = p(t_2) - c_d \eta(N(t_2)).\]
Since $p(t_0) = p(t_2)$, $N(t_0) = N(t_2)$, and $c_d (t_1 - t_0) \ge 0$, we have
\[\utilde(t_1) < p(t_2) - c_d \eta(N(t_2)) + c_d (t_1 - t_0).\]
Then $\utilde(t_0) = \utilde(t_1) - c_d (t_1 - t_0) < p(t_2) - c_d \eta(N(t_2)) = \utilde(t_2) = \utilde(t_0)$, a contradiction.
This implies that the optimal continuation payoff $u(t)$ cannot satisfy $p(t_2) - c_d \eta(N(t_2)) -u(t_2) = 0$, i.e. $g(t_2) = 0$.

  \paragraph{Case 2: $g(t_2) < 0$.}
  We now consider a continuous function $\uhat(\cdot)$ such that $\uhat$ satisfies \eqref{eq:-du}, and that
  \[\uhat(t_2) > p(t_2) - c_d \eta(N(t_2)) = \utilde(t_2),\]
  meaning that the corresponding $\ghat(t) \triangleq p(t) - c_d \eta(n(t)) -\uhat(t)$ satisfies $\ghat(t_2) < 0$.
  Then, for all $t \in [t_1, t_2]$, $\uhat(t) > \utilde(t)$ because the two functions %
  must not cross; otherwise, since they follow the same differential equation, their derivatives at the intersection point are the same, and thus they have the same value and derivative from then on until $t_2$.
  Therefore, by \eqref{eq:-du}, for all $t \in (t_1, t_2)$,
  \[- \frac\dd{\dd t} \uhat(t) < - \frac{\dd}{\dd t} \utilde(t).\]
  Taking the integration from $t_1$ to $t_2$, we have
  \[\uhat(t_1) - \uhat(t_2) < \utilde(t_1) - \utilde(t_2)\]
  and thus
  \[\uhat(t_0) - \uhat(t_2) < \utilde(t_0) - \utilde(t_2) < 0.\]
  Since $u(t_0) = u(t_2)$, it cannot be the case that $u(t_2) > p(t_2) - c_d \eta(N(t_2))$, i.e., $g(t_2) < 0$.
  
  \smallskip
  
  Since both cases lead to contradictions, we know that $g(t_2) > 0$, which completes the proof.

\section{The Examples are Stable Cycles} \label{appx:proof_prop_cycles_are_stable}

To apply \Cref{thm:suf_cond} we only need to verify conditions \ref{cond:repeat}, \eqref{eq:d(p_d)<c_d}, and \eqref{eq:strange_condition}.

Condition \ref{cond:repeat} staets that, within each period, the numbers of drivers arrived and dispatched are the same. During $t \in [t_0, t_1]$, all drivers are offline, so there are no dispatches. During $t \in [t_1, t_2]$, with $c_r = 0$, the number of dispatches drivers is
\[\int_{t_1}^{t_2} \lambda_r \bar F(p(t)) \dd t.\]
For the symmetric cycle, the total driver arrival is $\lambda_d T = 1 \times 12 = 12$, and the total rider request is
\[\int_{t_1}^{t_2} 5 \bar F(0.9 - 0.1 (t - t_1)) \dd t = 5 \int_6^{12} (0.1 t - 0.5) \dd t = 12.\]
For the asymmetric cycle, the total driver arrival is $\lambda_d T = 1 \times 6 = 6$, the total rider request is
\[\int_{t_1}^{t_2} 5 \bar F(0.9 - 0.1 (t - t_1)) \dd t = 5 \int_2^6 (0.1 t - 0.1) \dd t = 6.\]

To verify Condition \eqref{eq:d(p_d)<c_d}, we show that $\frac{\dd}{\dd t} (p(t) - c_d \eta(n(t))) < 0 \le c_d$.
For the symmetric cycle,
\begin{align*}
  p(t) - c_d \eta(n(t)) & = \frac{15 - t}{10} - \frac{9 \sqrt3}{200 \sqrt{8 + 14 t - t^2}}, \\
  \frac{\dd}{\dd t} (p(t) - c_d \eta(n(t))) & = -\frac1{10} - \frac{9 \sqrt3 (t - 7)}{200 (57 - (t - 7)^2)^{3/2}}, \\
  \frac{\dd^2}{\dd t^2} (p(t) - c_d \eta(n(t))) & = -\frac{9 \sqrt3 (57 + 2 (t - 7)^2)}{200 (57 - (t - 7)^2)^{5/2}} < 0 \text{ for } 6 < t < 12,
\end{align*}
i.e., $\frac{\dd}{\dd t} (p(t) - c_d \eta(n(t)))$ is decreasing, so we only need to check at $t = 6$: $\left. \frac{\dd}{\dd t} (p(t) - c_d \eta(n(t))) \right|_{t=6} = -0.0998 < 0 \le c_d$.
For the asymmetric cycle,
\begin{align*}
  p(t) - c_d \eta(n(t)) & = \frac{11 - t}{10} - \frac8{25 \sqrt{10 + 6 t - t^2}}, \\
  \frac{\dd}{\dd t} (p(t) - c_d \eta(n(t))) & = -\frac1{10} - \frac{8 (t - 3)}{25 (19 - (t - 3)^2)^{3/2}}, \\
  \frac{\dd^2}{\dd t^2} (p(t) - c_d \eta(n(t))) & = -\frac{8 (19 + 2 (t - 3)^2)}{25 (19 - (t - 3)^2)^{5/2}} < 0 \text{ for } 2 < t < 6,
\end{align*}
i.e., $\frac{\dd}{\dd t} (p(t) - c_d \eta(n(t)))$ is decreasing, so we only need to check at $t = 2$: $\left. \frac{\dd}{\dd t} (p(t) - c_d \eta(n(t))) \right|_{t=2} = -0.0958 < 0 \le c_d$.

Condition \eqref{eq:strange_condition} requires that %
that the maximum of $h(t)$ among all $t \in (t_1, t_2)$ is lower than %
the net payoff from a trip at time $t_0$ plus the cost of waiting from $t_0$ to $t_1$: $c_d (t_1 - t_0)$.
For the symmetric cycle illustrated in \Cref{fig:stable_cycle_1}, we have
\[p(t_0) - c_d \eta(N(t_0)) = 0.2862, \quad c_d (t_1 - t_0) = 0.18, \quad \max_{t \in (t_1, t_2)} h(t) = 0.4097 < 0.2862 + 0.18.\]
For the asymmetric cycle illustrated in \Cref{fig:stable_cycle_2},
\[p(t_0) - c_d \eta(N(t_0)) = 0.3988, \quad c_d (t_1 - t_0) = 0.16, \quad \max_{t \in (t_1, t_2)} h(t) = 0.3846 < 0.3988 + 0.16.\]

\section{Statement and Proof of Lemma \ref{prop:FFFF}}
\label{appx:FFFF}

In this section we prove a technical lemma used in the proof of \Cref{thm:lower_payoff}.

\begin{lemma} \label{prop:FFFF}
  If $\bar F(v)^{-1}$ is convex, i.e., for all $v$,
  \[\frac{\dd^2}{\dd v^2} \frac{1}{\bar F(v)} \ge 0,\]
  then for all $t_1$, $t_2$, and $p_r(\cdot)$,
  \[\frac{\int_{t_1}^{t_2} \bar F(p_r(t)) \dd t}{t_2 - t_1} \le \bar F \left(\frac{\int_{t_1}^{t_2} \bar F(p_r(t)) p_r(t) \dd t}{\int_{t_1}^{t_2} \bar F(p_r(t)) \dd t}\right).\]
\end{lemma}

\begin{proof}
  If $t_2 \to t_1$, %
  \begin{equation} \label{eq:FFFFv}
    \begin{split}
      & \quad \frac{t_2 - t_1}{\int_{t_1}^{t_2} \bar F(p_r(t)) \dd t} - \bar F \left(\frac{\int_{t_1}^{t_2} \bar F(p_r(t)) p_r(t) \dd t}{\int_{t_1}^{t_2} \bar F(p_r(t)) \dd t}\right)^{-1} \\
      & = \frac{2 \bar F(p_r(t_1))^2 - \bar F(p_r(t_1)) \bar F''(p_r(t_1))}{24 \bar F(p_r(t_1))^3} p_r'(t_1)^2 (t_2 - t_1)^2 + O((t_2 - t_1)^3) \\
      & = \left.\frac{\dd^2}{\dd v^2} \frac{1}{\bar F(v)}\right|_{v = p_r(t_1)} \frac{1}{24} p_r'(t_1)^2 (t_2 - t_1)^2 + O((t_2 - t_1)^3) \\
      & \ge O((t_2 - t_1)^3).
    \end{split}
  \end{equation}
  
  For the general case of $t_2$ and $t_1$, let $t_1 = x_0 < x_1 < \dots < x_{m-1} < x_m = t_2$, where $x_k = \frac{m - k}{m} t_1 + \frac{k}{m} t_2$. Let $m \to \infty$.
  \begin{align*}
    \frac{t_2 - t_1}{\int_{t_1}^{t_2} \bar F(p_r(t)) \dd t}
    & = \sum_{k=0}^{m-1} \frac{\int_{x_k}^{x_{k+1}} \bar F(p_r(t)) \dd t}{\int_{t_1}^{t_2} \bar F(p_r(t)) \dd t} \frac{x_{k+1} - x_k}{\int_{x_k}^{x_{k+1}} \bar F(p_r(t)) \dd t} \\
  \end{align*}
  By convexity of $\bar F(\cdot)^{-1}$,
  \begin{align*}
    \bar F \left(\frac{\int_{t_1}^{t_2} \bar F(p_r(t)) p_r(t) \dd t}{\int_{t_1}^{t_2} \bar F(p_r(t)) \dd t}\right)^{-1}
    & = \bar F \left(\sum\limits_{k=0}^{m-1} \frac{\int_{x_k}^{x_{k+1}} \bar F(p_r(t)) \dd t}{\int_{t_1}^{t_2} \bar F(p_r(t)) \dd t} \frac{\int_{x_k}^{x_{k+1}} \bar F(p_r(t)) p_r(t) \dd t}{\int_{x_k}^{x_{k+1}} \bar F(p_r(t)) \dd t}\right)^{-1} \\
    & \le \sum_{k=0}^{m-1} \frac{\int_{x_k}^{x_{k+1}} \bar F(p_r(t)) \dd t}{\int_{t_1}^{t_2} \bar F(p_r(t)) \dd t} \bar F \left( \frac{\int_{x_k}^{x_{k+1}} \bar F(p_r(t)) p_r(t) \dd t}{\int_{x_k}^{x_{k+1}} \bar F(p_r(t)) \dd t}\right)^{-1}.
  \end{align*}
  Taking the difference of the above two equations,
  \begin{align*}
    & \quad \frac{t_2 - t_1}{\int_{t_1}^{t_2} \bar F(p_r(t)) \dd t} - \bar F \left(\frac{\int_{t_1}^{t_2} \bar F(p_r(t)) p_r(t) \dd t}{\int_{t_1}^{t_2} \bar F(p_r(t)) \dd t}\right)^{-1} \\
    & \ge \sum_{k=0}^{m-1} \frac{\int_{x_k}^{x_{k+1}} \bar F(p_r(t)) \dd t}{\int_{t_1}^{t_2} \bar F(p_r(t)) \dd t} \left( \frac{x_{k+1} - x_k}{\int_{x_k}^{x_{k+1}} \bar F(p_r(t)) \dd t} - \bar F \left( \frac{\int_{x_k}^{x_{k+1}} \bar F(p_r(t)) p_r(t) \dd t}{\int_{x_k}^{x_{k+1}} \bar F(p_r(t)) \dd t}\right)^{-1} \right) \\
    & = \sum_{k=0}^{m-1} \frac{\int_{x_k}^{x_{k+1}} \bar F(p_r(t)) \dd t}{\int_{t_1}^{t_2} \bar F(p_r(t)) \dd t} O((x_{k+1} - x_k)^3) \tag*{by \eqref{eq:FFFFv}} \\
    & = 0 \quad \text{  as } m \to \infty.
  \end{align*}
  Finally, we take the inverse on both terms, completing the proof.
\end{proof}

\section{An Example}

\begin{example} \label{exmp:cycle_real}
Consider a %
setting where the arrival rates of drivers and riders are $\lambda_d = \SI{3}{\per\minute}$ and $\lambda_r = \SI{12}{\per\minute}$. Drivers' and riders' costs of time are $c_d = \$ \SI{20}{\per\hour} = \$ \frac13 \, \si{\per\minute}$ and $c_r = \$ 0 \, \si{\per\minute}$.
The pick-up time is parametrized by $\tau = \SI{15}{\minute}$ and $\alpha = 1/3$, so $\eta(n) = \tau n^{-\alpha} = 15 n^{-1/3} \, \si{\minute}$. Riders' value follows an exponential distribution with c.d.f.\ $F(v) = 1 - e^{-v / \$ 30}$. With the above parameters, the SOSS is given by $n^* = 15^{3/4} = 7.622$ and $p^* = \$ 30 \ln(4) = \$ 41.59$.
We further suppose the price change rate is restricted by $\ell_- = \ell_+ = \$ \SI{2}{\per\minute}$. 
  
Price cycles exist in this setting, and the following is one of the infinitely many. Let $t_0 = \SI{0}{\minute}$, $t_1 = \SI{10}{\minute}$, and $t_2 = \SI{20}{\minute}$. Thus, both the increasing and decreasing phases take $10$ minutes. Let $\bar p = \$ 30 \ln(3 (e^{2/3} - 1)) = \$ 31.35$, $\ubar p = \bar p - \$ 20 = \$ 11.35$, and $n(t_1) = 75$.
\Cref{fig:cycle_real} shows this cycle.  Note that the price in this cycle is always below the price in SOSS, which provides a definitive illustration of lower driver welfare in a price cycle.

\begin{figure}[t!]
  \centering
  \subcaptionbox{Contour plot of $\dot n(t)$ and the SOSS.}{\quad\includegraphics[height = 0.32 \textwidth]{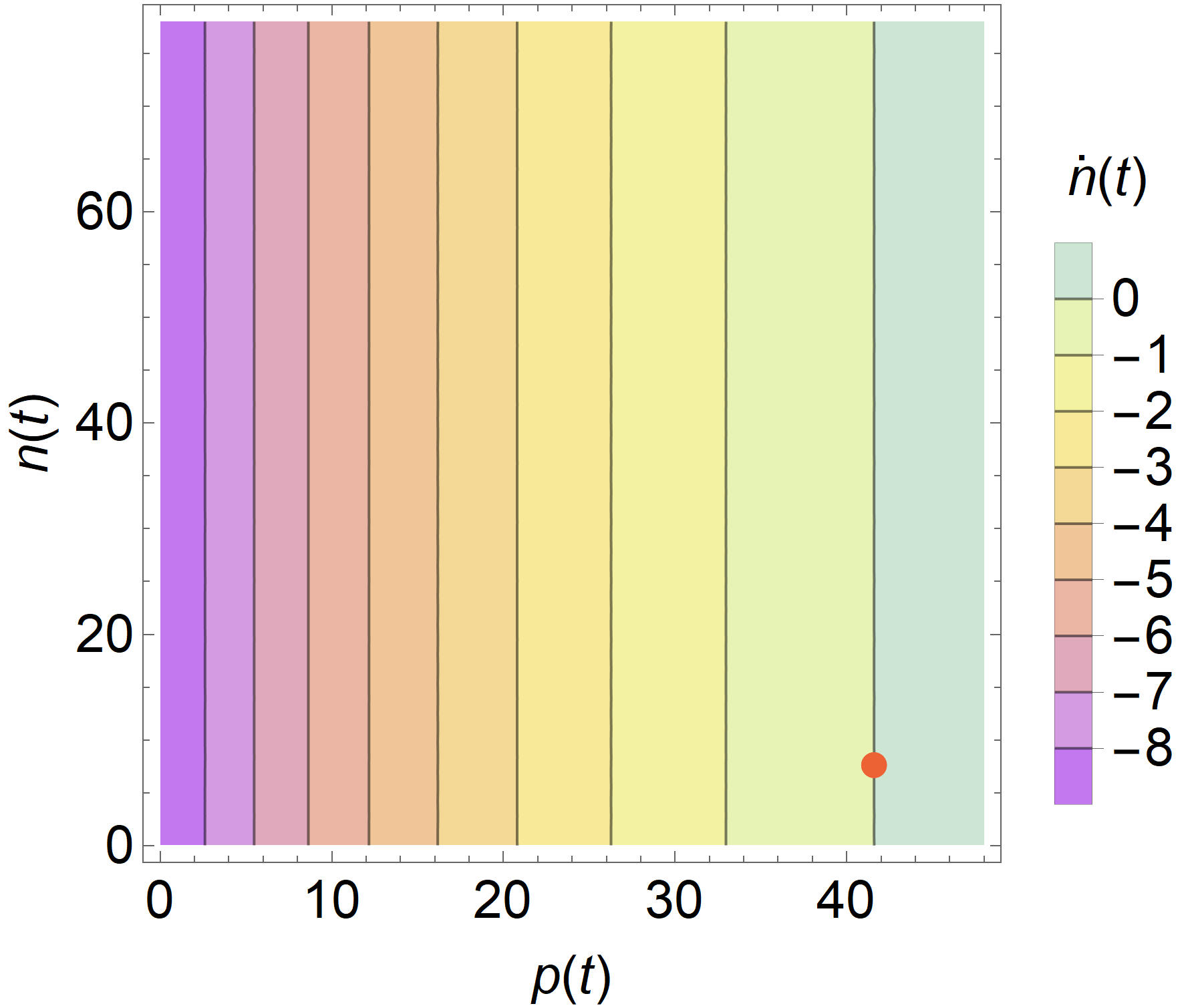}\quad}
  \hspace{3em}
  \subcaptionbox{Pricing policy and a stable cycle.}{\quad\includegraphics[height = 0.32 \textwidth]{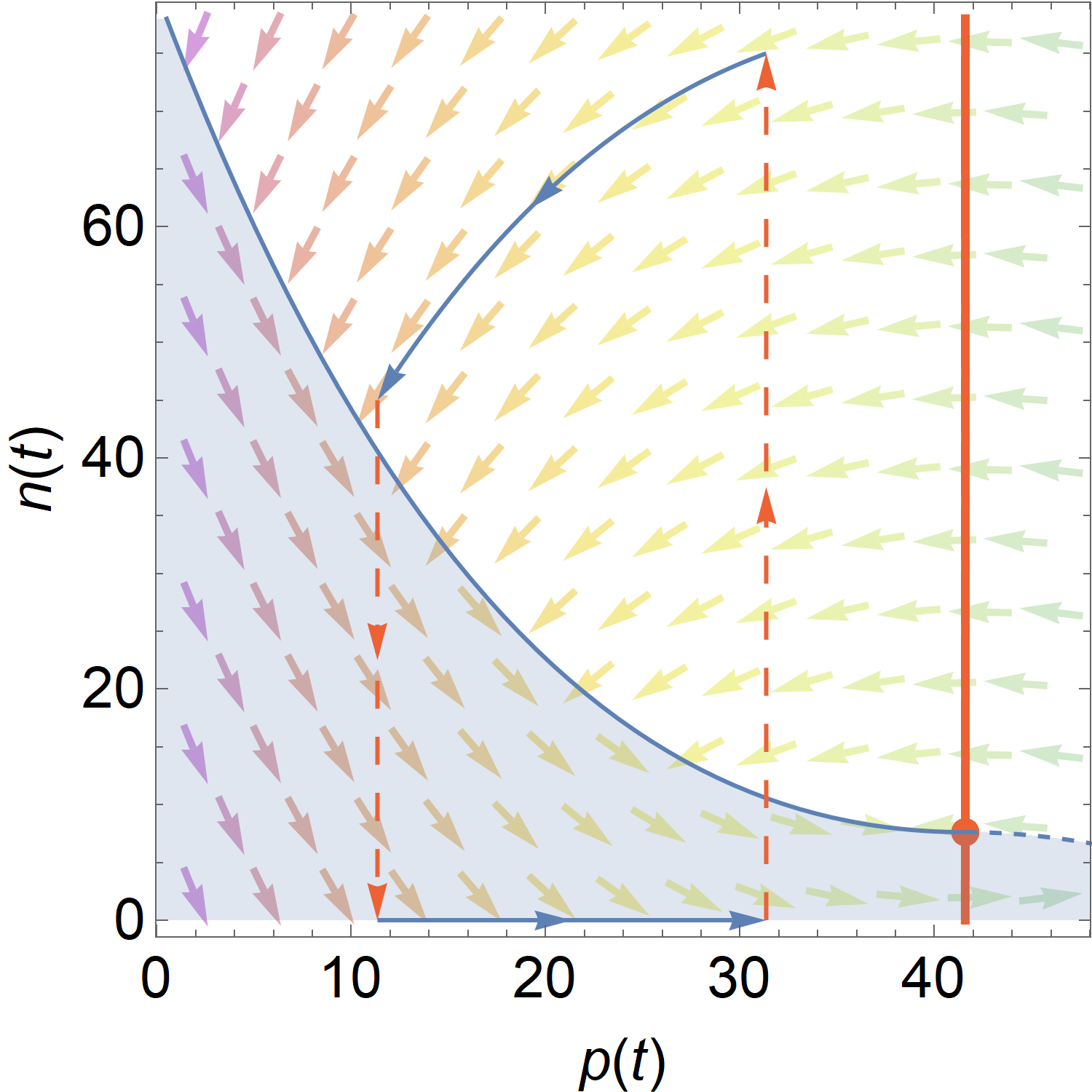}\quad}
  \caption{Illustration of the cycle discussed in \Cref{exmp:cycle_real}.
  }
  \label{fig:cycle_real}
\end{figure}
\end{example}

\section{The Impact of Smoothing Prices} \label{appx:prop_max_price_range}

In this section, we show that, when a platform requires that the prices decrease more ``smoothly'' (i.e. reducing the rate $\ell_-$ at which prices may decrease), the maximum price range $\bar p - \bar p$ decreases, meaning that stable price cycles become tend to become ``narrower''.  This highlights that smoothing prices reduces, but does not eliminate, price cycles.  However, smoothing prices too far as the drawback of making prices less responsive to market conditions, and so other approaches for eliminating price cycles are needed.

\begin{proposition} \label{prop:max_price_range}
Assume $c_r = 0$ and $V \sim \Unif[0, 1]$. The maximum price range $\bar{p} - \ubar{p}$ of any stable price cycle is $O(\sqrt{\ell_-})$.
\end{proposition}

\begin{proof}

When $c_r = 0$,
$\dot n(t) = \lambda_d - \lambda_r (1 - p(t))$.
Let $\hat{t} \in (t_1, t_2)$ denote the time such that $p(\hat{t}) = p^\ast = \bar{F}^{-1}(\lambda_d / \lambda_r)$. We know, $\dot n(t) > 0$ for $t \in (t_1, \hat{t})$,  $\dot n(t) < 0$ for $t \in (\hat{t}, t_2)$, and $n(\hat{t}) = \hat{n}$. At time $t_2$, the total number of drivers in the region $N(t_2)$ can be written as:

\begin{align}
  N(t_2) 
  & = \hat n + \int_{\hat{t}}^{t_2} \dot n(t) \dd t \notag \\
  & = \hat n + \int_{\hat{t}}^{t_2} \left( \lambda_d - \lambda_r + \lambda_r p(t) \right) \dd t \notag \\
  & = \hat n - \frac{1}{\ell_-} \int_{p^\ast}^{\ubar p} \left( \lambda_d - \lambda_r + \lambda_r p \right) \dd p  \label{eq:n_dt_dp} \\ 
  & = \hat n - \frac{1}{2 \ell_-} \left(1 - \frac{\lambda_d}{\lambda_r} - \ubar p\right)^2. \label{eq:eliminate_p^*}
\end{align}

Equation \eqref{eq:n_dt_dp} follows from the fact that $\dd p = - \ell_- \dd t$ %
for $t \in (t_1, t_2)$,
and \eqref{eq:eliminate_p^*} follows from $p^* = \bar{F}^{-1}(\lambda_d / \lambda_r) = 1 - \lambda_d/\lambda_r$.
Similarly, the number of drivers in the region at time $t_1$ satisfies:

\begin{align}
    N(t_1) = \hat n - \frac{1}{2 \ell_-} \left(1 - \frac{\lambda_d}{\lambda_r} - \bar p\right)^2.
\end{align}

$N(t_1)$ and $N(t_2)$ need to be non-negative. This implies 
\begin{align}
    \ubar p & \geq 1 - \lambda_d / \lambda_r - \sqrt{\frac{2 \hat{n} \ell_- }{\lambda_r}}, \\
    \bar p & \leq 1 - \lambda_d / \lambda_r + \sqrt{\frac{2 \hat{n} \ell_- }{\lambda_r}}.
\end{align}
As a result, when the maximum number of drivers $\hat{n}$ is bounded, the price range $\bar p - \ubar p $ satisfies:
\begin{align}
  \bar p - \ubar p \leq \sqrt{\frac{8 \hat{n} \ell_- }{\lambda_r}} = O(\sqrt{\ell_-}).
\end{align}

\end{proof}

\end{document}